\documentclass[12pt]{article}
\usepackage[T1]{fontenc}
\usepackage[utf8]{inputenc}
\usepackage[letterpaper]{geometry}
\usepackage{color}
\usepackage{mathtools}
\usepackage{amsmath}
\usepackage{amsthm}
\usepackage{amssymb}
\usepackage{graphicx}
\usepackage{setspace}
\usepackage[authoryear]{natbib}
\usepackage[unicode=true,
  bookmarks=false,
breaklinks=false,pdfborder={0 0 1},backref=section,colorlinks=true]
{hyperref}
\hypersetup{
linkcolor=blue,citecolor=blue,urlcolor=blue}

\makeatletter
\providecommand{\tabularnewline}{\\}

\usepackage{bm}
\usepackage{enumitem}
\usepackage{booktabs}

\newtheorem{assumption}{Assumption}\newtheorem{proposition}{Proposition}\newtheorem{theorem}{Theorem}\newtheorem{remark}{Remark}\newtheorem{definition}{Definition}\newtheorem{example}{Example}

\newcommand{\E}{\mathbb{E}}
\newcommand{\Prob}{\mathbb{P}}
\newcommand{\R}{\mathbb{R}}
\newcommand{\ind}{\mathbf{1}}

\newcommand{\N}{\mathcal{N}}

\makeatother

\begin{document}
\title{Tractable Identification of Strategic Network Formation Models with
Unobserved Heterogeneity\footnote{We thank Xiaohong Chen, Bryan Graham, Roger Moon, Elie Tamer, and Rui Wang for helpful comments and suggestions.}}
\author{Wayne Yuan Gao\thanks{Department of Economics, University of Pennsylvania. Email: \texttt{waynegao@upenn.edu}}
  \and Ming Li\thanks{Department of Economics and Risk Management Institute, National University of Singapore. Email:
\texttt{mli@nus.edu.sg}} \and Zhengyan Xu\thanks{Department of Economics, University of Pennsylvania. Email: \texttt{xzy24@sas.upenn.edu}} }

\date{\today}
\maketitle
\begin{abstract}
  \noindent We develop a tractable identification approach for strategic network formation models with both strategic link interdependence and individual unobserved heterogeneity (fixed effects). The key challenge is that endogenous network statistics (e.g.\ number of common friends) enter the link formation equation, while the mapping from model primitives to equilibrium network structure is generally intractable. Our approach sidesteps this difficulty using a ``bounding-by-$c$'' technique that treats endogenous covariates as random variables and exploits monotonicity restrictions to obtain identifying information. A central contribution is to develop a spectrum of fixed-effects handling strategies based on subnetwork configurations: tetrad-based restrictions that \emph{difference out} all individual fixed effects, triad-based and weighted restrictions that combine \emph{difference-out} and \emph{integrate-out} steps by differencing out some fixed effects and profiling over the remainder conditional on observed characteristics, and general weighted cycle-based restrictions that unify these cases. We also provide point identification results. Preliminary simulations show that the approach can deliver informative bounds on the structural parameters.

  \bigskip{}

\end{abstract}
\newpage{}


\section{Introduction}

Network formation is a central problem in network
economics. From the econometric perspective, the goal
is to identify and estimate structural parameters
in a model where agents are heterogeneous and strategically interdependent
when making linking decisions. Substantial progress has been
made on models with either unobserved heterogeneity or strategic interdependence, but
combining both features in a tractable framework has remained an open
problem.

This paper develops a tractable identification approach for strategic
network formation models with unobserved individual fixed effects.
We consider a latent utility framework where the formation of a link
between agents $i$ and $j$ depends on observed dyadic characteristics,
individual fixed effects that capture unobserved degree heterogeneity
(such as sociability or popularity), and idiosyncratic pairwise shocks.
The link formation decision is allowed to depend on endogenous
network statistics arising from the equilibrium of the network
formation game, such as the number of common friends between $i$
and $j$. This accommodates strategic complementarities and
other forms of link interdependence that are relevant in many economic
applications.

The main methodological challenge is that the mapping
from model primitives to equilibrium network structure is generally
intractable. In strategic network formation games, the equilibrium
network depends on the entire profile of agent characteristics, fixed
effects, and shock realizations through complex equilibrium conditions.
This gives rise to at least three difficulties. First, the space of possible
network structures is a discrete space with a combinatorially overwhelming
number of elements\footnote{For a standard illustrative example: there are $2^{435}$ possible
(undirected and unweighted) network structures on a set of 30 agents.}, making it very hard to handle both theoretically and computationally.
Second, solution concepts for network formation problems (such as
pairwise stability) usually do not have uniqueness properties, and
sometimes the number of equilibria can be very large \citep{de2020econometric}. Third, iterative
procedures are in general not guaranteed to converge to an equilibrium
network even when such an equilibrium does exist \citep{jackson2002evolution}. Consequently, characterizing
the equilibrium mapping, let alone inverting it for identification
purposes, is typically infeasible, except in very simple special cases.

This intractability has led the existing literature to either (i)
focus on models without strategic interdependence, such as in the seminal work by \cite{graham2017econometric} and follow-up work by \cite{gao2020} and \citet*{gao2023logical},
or (ii) study strategic models under restrictions that eliminate or
simplify the role of unobserved heterogeneity \citep*{mele2017structural,de2018identifying,sheng2020structural}.

Our approach avoids the need for equilibrium characterization
by exploiting monotonicity restrictions to obtain identifying information
without knowing the form of
the equilibrium mapping. The technique, which we call ``bounding
by $c$,'' treats endogenous covariates as random variables and uses
indicator function arguments to derive bounds on structural parameters
that hold regardless of the equilibrium realization. This
approach builds on and extends ideas from \citet{gao2026identification},
who developed similar techniques for dynamic panel data models under
a partial stationarity condition.

Based on our key technique, we derive a system of identifying restrictions
based on different subnetwork configurations. Our primary restrictions
exploit tetrad configurations, i.e., sets of four agents with a particular
pattern of links, to difference out all individual fixed effects simultaneously.
The resulting bounds depend only on the distribution of idiosyncratic
shocks, which can be aggregated into identifying restrictions
on the model parameters under a standard independence assumption on
the idiosyncratic shocks. We also complement the tetrad restrictions
with two additional types of restrictions: (1) ``incomplete differencing''
restrictions, such as triad-based restrictions, that do not eliminate
all fixed effects but nevertheless provide identifying information,
and (2) ``cyclically differencing'' restrictions that involve longer
cycles of links than those in a tetrad.

Beyond partial identification, we establish conditions under which the structural parameters are point identified. When the pairwise shocks follow a logistic distribution and the endogenous covariates satisfy a comparison-pattern invariance condition---which holds for both the common-friends statistic and the standard Jaccard index---we show that a log-odds ratio of tetrad probabilities, conditioned on diagonal link absences that isolate the endogenous covariates from the tetrad shocks, identifies a linear index of the structural parameters. The fixed effects cancel algebraically through tetrad differencing, while the endogenous covariates decouple through the isolation conditioning; these two mechanisms operate separately. The resulting identification equation leads to a computationally simple conditional logit estimator that accommodates both endogenous network statistics and unobserved individual fixed effects, generalizing the tetrad logit of \citet{graham2017econometric} to settings with strategic link interdependence.

We report a preliminary simulation exercise illustrating the finite-sample performance of the proposed identifying restrictions. We simulate networks under nested specifications that progressively add fixed effects and an endogenous covariate capturing local network structure, and compute the resulting identified set using the tetrad inequalities. The results show that the restrictions can deliver nontrivial bounds on the strategic-interaction parameter even in the full model with endogenous covariates and fixed effects. A larger-scale Monte Carlo study is left to future work.

\subsubsection*{Related Literature}

This paper relates to several strands of the econometric literature
on network formation; see \cite{de2020econometric} for a recent survey. The first strand concerns dyadic models with
unobserved heterogeneity but without strategic interdependence. \citet{graham2017econometric}
introduced the tetrad logit estimator, which differences out additive
fixed effects by comparing link patterns across tetrads of agents.
This approach achieves point identification of homophily parameters
under a conditional logit specification. \citet{gao2020}
extended these ideas to nonparametric identification of the homophily
effect function, while \citet*{gao2023logical} developed logical differencing
techniques for settings with non-transferable utilities where bilateral
consent is required for link formation. Our paper builds on these
ideas but generalizes the setting by allowing
endogenous covariates arising from strategic interaction.

A complementary line of recent work develops general fixed-effects methods applicable to network data. \citet{bonhomme2023functional} extend the functional differencing approach of \citet{bonhomme2012functional} from panel data to network settings, deriving moment restrictions on model parameters that hold regardless of the form of heterogeneity and without requiring dense networks; their framework, however, treats the network as exogenous. \citet*{bonhomme2025moment} characterize all moment conditions for nonlinear panel data models that are robust to both unrestricted feedback and arbitrary heterogeneity, with applications to duration and count models. \citet*{dano2025binary} provide a systematic treatment of conditional likelihood and moment-based identification in binary logit models with general fixed effects, covering both panel and network (dyadic) data; their analysis subsumes the tetrad logit of \citet{graham2017econometric} as a special case. Our paper differs from this body of work in that we address the additional complication of endogenous network covariates arising from strategic interaction, which requires a distinct identification strategy based on tetrad inequalities rather than conditional likelihood or moment equalities.

The second strand studies strategic network formation under various
equilibrium concepts. \citet{mele2017structural} analyzes exponential
random graph models (ERGMs) as potential games with strategic complementarities,
developing simulation-based estimation methods. However, ERGMs are
known to suffer from degeneracy problems and computational challenges
in large networks. \citet*{de2018identifying} and \citet{sheng2020structural}
study partial identification in strategic network formation models
under simultaneity, using subnetwork restrictions to bound the identified
set. \citet{sheng2020structural} provides particularly elegant results
using small subnetwork configurations, but her analysis does not accommodate
agent-level fixed effects. \citet{menzel2026strategic} develops a many-agent asymptotic approximation for pairwise stable networks with anonymous and non-anonymous interaction effects, characterizing the limiting distribution of link intensities through a fixed-point system. His identification of payoff parameters requires parametric specification of the utility function and solution of the equilibrium fixed point, and does not address individual fixed effects of the type we consider here. Our approach addresses the fixed-effects gap by combining
the differencing logic from the dyadic literature with
techniques that accommodate strategic interdependence, while avoiding the need for equilibrium computation.

On the inference side, \citet{leung2019treatment} and \citet{leung2019normal} develop laws of large numbers and central limit theorems for network moments in sparse strategic formation models embedded in a latent position space, using a branching-process subcriticality condition to control strategic dependencies. \citet{menzel2021clt} provides an alternative central limit theorem based on the exchangeability of agents' potential values, which does not require positional homophily or $K$-locality and applies to general $D$-adic network moments. In Section~\ref{sec:primitive}, we provide primitive conditions under which the conditional probabilities underlying our identifying restrictions can be consistently estimated from a single large network, drawing on these frameworks.

To our knowledge, this paper is the first in the econometric literature
to provide identification results for network formation models that
allow for both (i) unobserved individual fixed effects
and (ii) endogenous network statistics arising from strategic link
formation. Our contributions are as follows. First, we show that
tetrad-based configurations can be used to difference out all individual
fixed effects even when endogenous covariates are present, yielding
bounds on structural parameters that are informative under reasonable
support conditions. Second, we develop complementary ``incomplete
differencing'' restrictions that deliver additional identification
power by combining the two traditional ``difference-out'' and ``aggregate-out''
approaches for the handling of fixed effects in the econometric literature.
Third, we demonstrate that the ``bounding by $c$'' technique from
the panel data literature in \citet{gao2026identification} can be productively
adapted to network settings, suggesting broader applicability of the
key idea.

Methodologically, our paper builds upon the recent work on semiparametric
identification with non-separable unobservables. The ``bounding by
$c$'' technique we employ shares conceptual foundations with the
partial stationarity approach of \citet{gao2026identification} for dynamic
panel models, and the ``multi-inequality aggregation'' also relates
to the logical operations underlying \citet*{gao2023logical}. The
commonality across these papers is the use of arithmetic and logical
operations based on monotonicity and inequality conditions. In particular,
the ``bounding by $c$'' technique in \citet{gao2026identification} and
the present paper provides a way to ``eliminate'' the potentially complicated
endogenous variables when establishing identifying restrictions based
on homogeneity assumptions on the structural error
terms.

Finally, our paper relates to the broader literature on games with
incomplete information and heterogeneous agents. \citet{aguirregabiria2007sequential}
and \citet*{bajari2007estimating} develop estimation methods for dynamic
games that avoid full solution of the game by using conditional choice
probability representations. While our setting differs (we study
a static network formation game rather than a dynamic game), the
spirit of avoiding intractable equilibrium computations through the use of observable implications is similar. That said, the exact techniques we exploit to achieve this in our network formation context are naturally different from theirs.

\medskip{}

The remainder of the paper is organized as follows. Section~\ref{sec:model}
presents the model setup, introduces notation, and states our maintained
assumptions. Section~\ref{sec:identification} derives the main identifying
restrictions, beginning with tetrad-based restrictions, then developing
triad-based and weighted alternatives, and culminating in a general
cycle-based differencing framework. Section~\ref{sec:primitive} provides primitive sufficient conditions under which the high-level identifying assumptions hold, embedding the model in the sparse network framework of \citet{leung2019treatment}. Section~\ref{sec:point_id} establishes conditions for point identification under logistic errors and a comparison-pattern invariance condition on the endogenous covariates. Section~\ref{sec:simulation}
provides simulation results based on tetrad-based identifying restrictions.
Section~\ref{sec:discussion} concludes.


\section{Model and Assumptions} \label{sec:model}

\subsection{Model Setup\label{subsec:setup}}

Consider a set of $n$ agents indexed by $i=1,\ldots,n$, and an
unweighted network among them represented by an $n\times n$ adjacency matrix
$Y$, where $Y_{ij}=1$ if a link exists between agents $i$ and $j$,
and $Y_{ij}=0$ otherwise. We focus on undirected networks, i.e., $Y_{ij}=Y_{ji}$
for all $i,j$.  Throughout, we index distinct unordered pairs as
$(i,j)$ with $i<j$; all dyadic quantities ($Y_{ij}$, $Z_{ij}$,
$X_{ij}$, $\varepsilon_{ij}$) are understood to be symmetric in
their indices.

We consider the following network formation model with both link interdependence and fixed effects, where a link between agents $i$ and $j$ exists if and only if the latent
surplus from the link is non-negative:
\begin{equation}
  Y_{ij}=\ind\left\{ Z_{ij}'\beta_{0}+X_{ij}'\gamma_{0}\geq A_{i}+A_{j}+\varepsilon_{ij}\right\} \label{eq:link_formation}
\end{equation}
where:
\begin{itemize}
  \item $Z_{ij}=w_n(Z_{i},Z_{j})$ is a vector of exogenous covariates
    constructed from individual-level exogenous characteristics $Z_{i}$
    and $Z_{j}$ via a known function $w_n$ (e.g., $Z_{ij}=|Z_{i}-Z_{j}|$
    for homophily effects);
  \item $X_{ij}$ is a vector of potentially endogenous covariates, which may involve other links $Y_{hk}$ with $(h,k)\neq(i,j)$: we discuss $X_{ij}$ in more detail below.
  \item $A_{i}$ and $A_{j}$ are individual-level unobserved fixed effects
    capturing heterogeneity in sociability, popularity, or other unobserved
    characteristics (here $A_i$ acts as a threshold, so larger $A_i$ reduces the link probability; equivalently, one may place $-A_i$ on the left-hand side, recovering the usual interpretation where higher sociability promotes link formation);
  \item $\varepsilon_{ij}$ is an idiosyncratic pairwise shock.
\end{itemize}
The covariate function $w_n$ may depend on the network size $n$; this generality is essential for accommodating sparse network asymptotics (Section~\ref{sec:primitive}). Specifically, when agents are embedded in a latent space with positions $\Xi_i$ and the network is sparse with $O(1)$ expected degree, the sparsity scaling $r_n \to 0$ rescales positions to $\bar{\Xi}_i := r_n^{-1}\Xi_i$, and the covariate function takes the form $w_n(Z_i, Z_j) = \tilde{w}(\bar{\Xi}_i, \bar{\Xi}_j, W_i, W_j)$ for a \emph{fixed} function $\tilde{w}$ that does not depend on $n$. The $n$-dependence of $w_n$ thus enters only through the position rescaling, while the structural parameters $\theta_0 = (\beta_0, \gamma_0)$ remain fixed---analogous to the spatial weight matrix $W_n$ in spatial autoregressive models. This structure ensures that the conditional link probability $\Prob(Y_{ij} = 1 \mid Z_{ij} = z)$ is an $O(1)$ quantity for any fixed $z$ in the support of $Z_{ij}$, even though the marginal link probability vanishes under sparsity. In many specifications $w_n$ does not actually depend on $n$, but we allow it for generality.

\begin{remark}[Individual-Level versus Dyadic Rescaling]\label{rem:individual_rescaling}
  In the sparse network framework of \citet{leung2019treatment}, the
  sparsity scaling~$r_n = (\kappa/n)^{1/d_\xi}$ enters through the
  rescaled homophily measure
  $\delta_{ij} = r_n^{-1}\|\Xi_i - \Xi_j\|$, which is a function of
  the \emph{dyadic} distance.  We adopt an equivalent but notationally
  simpler formulation in which the rescaling is absorbed into the
  individual-level characteristics.  Specifically, we define the
  rescaled type
  \[
    \bar Z_i \;:=\; (r_n^{-1}\Xi_i,\; W_i)
  \]
  and write the dyadic covariate as
  $Z_{ij} = \tilde{w}(\bar Z_i, \bar Z_j)$ for a \emph{fixed}
  function~$\tilde{w}$ that does not depend on~$n$.  The
  $n$-dependence of the model is thereby localized in the
  deterministic, individual-level transformation
  $Z_i \mapsto \bar Z_i$, while the structural objects---the
  parameters $\theta_0 = (\beta_0, \gamma_0)$ and the covariate
  function~$\tilde{w}$---remain invariant to the network size.

  The two formulations are mathematically identical: writing
  $w_n(Z_i, Z_j) := \tilde{w}(\bar Z_i, \bar Z_j)$ recovers the
  size-dependent dyadic function used in \citet{leung2019treatment}.
  In particular, because $\bar Z_i$ is a deterministic function
  of~$Z_i$ and $(Z_i, A_i)$ is i.i.d.\ across agents
  (Assumption~\ref{ass:random_sampling}), the augmented type
  $\tilde Z_i := (\bar Z_i, A_i)$ remains i.i.d., and all arguments
  based on exchangeability and the branching-process domination in
  \citet{leung2019treatment} carry over without modification.
  Likewise, the identification results of
  Theorems~\ref{thm:tetrad}--\ref{thm:tetrad_parametric} are
  unaffected, since they are purely algebraic and hold for each
  fixed~$n$ regardless of how the $n$-dependence is parametrized.

  We prefer the individual-level formulation for two reasons.  First,
  it cleanly separates the sparsity mechanism (the rescaling
  $\Xi_i \mapsto r_n^{-1}\Xi_i$, which governs the rate at which
  distant agents become unable to interact) from the identification
  mechanism (the function~$\tilde{w}$ and the conditioning
  on~$\zeta_{ijhk}$, which are $n$-invariant).  Second, it makes
  transparent that the conditional link probability
  $\Prob(Y_{ij}=1 \mid Z_{ij} = z)$ is an $O(1)$ quantity for any
  fixed~$z$ in the support of~$Z_{ij}$: the vanishing marginal link
  probability $\Prob(Y_{ij}=1) = O(n^{-1})$ arises entirely from the
  thinning of the position distribution under rescaling, not from any
  $n$-dependence in the structural model.
\end{remark}

The structural parameters of interest are $\theta_{0}:=\left(\beta_{0},\gamma_{0}\right)\in \Theta$.
Throughout this paper, we use the shorthand notation:
\begin{equation}
  \delta_{ij} :=\delta_{ij}(\theta_0),\quad \delta_{ij}(\theta):=Z_{ij}^{'}\beta+X_{ij}^{'}\gamma,\quad v_{ij}:=A_{i}+A_{j}+\varepsilon_{ij}\label{eq:shorthand}
\end{equation}
so that the link formation equation becomes
\begin{equation}
  Y_{ij}=\ind\{v_{ij}\leq\delta_{ij}\}.
  \label{eq:link_shorthand}
\end{equation}


A distinguishing feature of our model is the presence of the endogenous covariates $X_{ij}$ and the fixed effects $A_i,A_j$. In particular, the endogenous covariates $X_{ij}$ may arise as functions of the realized network. Specifically,
we allow:
\begin{equation}
  X_{ij}=\phi_{ij}(Y,Z)\label{eq:endogenous_X}
\end{equation}
where $\phi_{ij}$ is a known function mapping the realized network
$Y$ and exogenous characteristics $Z=(Z_{1},\ldots,Z_{n})$ to the
dyadic covariate. For example, a component of $X_{ij}$ might be the
number of common friends between $i$ and $j$:
\begin{equation}
  \text{CF}_{ij}:=\sum_{k\neq i,j}Y_{ik}Y_{jk}\label{eq:common_friends}
\end{equation}
which captures transitivity effects in network formation. Other examples
include measures of local clustering, degree statistics, and their rescaled or normalized variants.

Since $X_{ij}$ depends on the entire network $Y$, and $Y$ is determined
by equation~\eqref{eq:link_formation} simultaneously for all pairs,
the model describes a strategic network formation game under transferable utilities with pairwise stability as the solution concept. In general, the game may have multiple equilibria for a given realization of primitives $(Z,A,\varepsilon)$, where $A=(A_{1},\ldots,A_{n})$ and $\varepsilon=(\varepsilon_{ij})_{i<j}$
collect the fixed effects and idiosyncratic shocks respectively. We represent the realized network abstractly as:
\begin{equation}
  Y=g(Z,A,\varepsilon;\theta_{0})\label{eq:equilibrium}
\end{equation}
where $g$ incorporates both the equilibrium correspondence and an (arbitrary and possibly unknown) equilibrium selection mechanism that picks a single equilibrium for each realization of the primitives.

The equilibrium mapping
$g$ is generally intractable to characterize. Even for simple specifications
of $\phi_{ij}$, the fixed-point nature of the equilibrium creates a complex interdependence
structure, which is exacerbated by the discrete combinatorial nature of the graph space. As will be shown below, our identification approach does not require characterization
of $g$: we extract identifying information using monotonicity
restrictions that hold regardless of the complexity of the equilibrium correspondence and any equilibrium selection mechanisms.


\subsection{Assumptions}


We maintain the following assumptions throughout the analysis.

\begin{assumption}[Random Sampling]\label{ass:random_sampling}
  The pair $(Z_{i},A_{i})$ is independently and identically distributed
  across agents $i=1,\ldots,n$.
\end{assumption}

Assumption \ref{ass:random_sampling}  is standard in the network formation literature. Note that arbitrary dependence between $Z_{i}$ and $A_{i}$ within
an agent is allowed, as is arbitrary dependence between $A_i$ and $\varepsilon_{ij}$. The latter flexibility is inessential for the tetrad-based restrictions in Section~\ref{sec:identification}, since the tetrad differencing construction eliminates all fixed effects algebraically, but it avoids imposing unnecessary restrictions on the model.

\begin{assumption}[Exogeneity of $Z_i$]\label{ass:exogeneity} The vector
  of exogenous characteristics $Z=(Z_{1},\ldots,Z_{n})$ is independent
  of the idiosyncratic shocks $\varepsilon=(\varepsilon_{ij})_{i<j}$.
\end{assumption}

Assumption~\ref{ass:exogeneity} is an exogeneity condition
requiring that the observables $Z$ are independent of the unobserved
pairwise shocks. Note that we do \emph{not} assume that $X$ is independent
of $\varepsilon$: since $X$ is a function of the equilibrium network,
it is endogenous and generally correlated with the idiosyncratic shocks.
The key distinction is between the exogenous covariates $Z$ (which
satisfy independence) and the endogenous covariates $X$ (which do
not).

\begin{assumption}[IID Pairwise Errors]\label{ass:iid_errors}
  The idiosyncratic shocks $\varepsilon_{ij}$ are independently and
  identically distributed across all pairs $(i,j)$ with $i<j$.
\end{assumption}

Assumption \ref{ass:iid_errors} is another standard assumption in the network formation literature on the idiosyncrasy of link-level surplus shocks.

\begin{assumption}[Consistent Estimation of Tetrad Conditional Probabilities]\label{ass:tetrad_id}
  Let $ijhk$ denote a tetrad of distinct agents and write $Z_{S}:=(Z_i,Z_j,Z_h,Z_k)$ for their individual-level exogenous characteristics. Let $E_{ijhk}$
  be an event in the observable tetrad $\sigma$-algebra
  \[
    \sigma_{ijhk}:=\sigma\left(Y_{i'j'},X_{i'j'},Z_{i'}:i',j'\in\{i,j,h,k\}\right).
  \]
  Then the conditional probability $\Prob(E_{ijhk}\mid Z_S=z_S)$
  can be consistently estimated from observable data in a single large network.
\end{assumption}

Assumption~\ref{ass:tetrad_id} is a high-level regularity condition
asserting that conditional probabilities of tetrad events can be consistently
estimated from a single large network. The conditioning is on the individual-level characteristics $(Z_i,Z_j,Z_h,Z_k)$, which is consistent with the general subnetwork formulation (Assumption~\ref{ass:W_id} below). Since the link formation equation~\eqref{eq:link_formation} depends on agents $i$ and $j$ only through the dyadic covariate $Z_{ij}=w_n(Z_i,Z_j)$, it is convenient to define the shorthand
\begin{equation}
  \zeta_{ijhk} := (Z_{ij}, Z_{hk}, Z_{ik}, Z_{jh})
  \label{eq:zeta_ijhk_app}
\end{equation}
for the vector of dyadic covariates determined by the tetrad. Because $\zeta_{ijhk}$ is a deterministic function of $Z_S$, any conditional probability $\Prob(E\mid \zeta_{ijhk}=\zeta)$ is well-defined under the finer conditioning of Assumption~\ref{ass:tetrad_id}, and we use the shorthand $\E[\cdot\mid\zeta]$ throughout Section~\ref{sec:identification} for readability.

This assumption is satisfied when the network exhibits sufficiently weak dependence so that a Law of Large Numbers (LLN) applies to tetrad statistics. We view Assumption~\ref{ass:tetrad_id} as a convenient high-level condition for presenting the identification arguments; Section~\ref{sec:primitive} provides primitive sufficient conditions under which it holds for the class of strategic formation games considered here, drawing on the network limit theory of \citet{leung2019treatment}.

\begin{remark}[Population Identification versus Single-Network Estimation]\label{rem:id_vs_est}
  The interplay between identification and estimation in network models differs from the standard cross-sectional setting, and it is useful to clarify the distinction.

  In classical econometrics with i.i.d.\ data, the data-generating process (DGP) is fixed, identification asks whether the structural parameters are pinned down by the population distribution, and estimation recovers them from a growing sample drawn from that same distribution. In our setting, the data consists of a \emph{single} network of $n$ agents, and the game-theoretic structure means that the joint distribution of $(Y, X, Z)$ is inherently $n$-dependent: adding agents changes the strategic environment and hence the equilibrium mapping~$g$ in~\eqref{eq:equilibrium}. When the covariate construction $w_n$ depends on $n$ (as required for sparse network asymptotics; see Section~\ref{sec:primitive}), the observables are defined relative to the current network size.

  Accordingly, it is useful to distinguish two layers. \emph{Population-level identification} is a per-$n$ statement: for a game of fixed size $n$, the bounding-by-$c$ restrictions derived in Section~\ref{sec:identification} hold whenever the conditional probabilities $p_L(\zeta,c;\theta)$ and $p_U(\zeta,c;\theta)$ are well-defined. This layer is purely algebraic---it relies on monotonicity and the exogeneity of $Z$ (Assumption~\ref{ass:exogeneity}), with no reference to asymptotics, sparsity, or limit theory. In particular, the identified set---which we denote $\Theta_I^{(n)}$ to emphasize its dependence on $n$ (see, e.g., $\Theta_I^{\mathrm{tetrad}}$ in Theorem~\ref{thm:tetrad})---is well-defined for each network size $n$, and the true parameter satisfies $\theta_0 \in \Theta_I^{(n)}$ for all $n$, regardless of whether the network is dense or sparse. \emph{Consistent estimation from a single network} is an asymptotic statement: as the network grows ($n \to \infty$), the sample conditional probabilities $\hat{p}(\zeta)$ converge to their population counterparts $p^{(n)}(\zeta)$. This convergence requires a law of large numbers for dependent network statistics, which in turn requires sufficient decay of dependence across the network. Assumption~\ref{ass:tetrad_id} asserts that such convergence holds, and Section~\ref{sec:primitive} provides primitive conditions---in particular, sparsity and subcriticality---under which it can be verified.

  The $n$-dependence of the covariate construction $w_n$ is relevant only at the estimation layer; the identification layer is agnostic about it. Under the regularity conditions of Section~\ref{sec:primitive}, the equilibrium ``localizes'' as $n \to \infty$---the conditional probability of a tetrad event depends only on the $O(1)$-sized strategic neighborhood---so that $p^{(n)}(\zeta)$ stabilizes and the identified sets $\Theta_I^{(n)}$ converge to a well-defined limit.
\end{remark}


\section{Identifying Restrictions}

\label{sec:identification} 

This section develops the main identifying restrictions for the structural
parameters $\theta_{0}=(\beta_{0},\gamma_{0})$. We begin with tetrad-based
restrictions that achieve complete elimination of all individual fixed
effects, then present additional triad-based restrictions as well as a general class of cycle-based ones.


\subsection{Tetrad-Based Restrictions}

\label{subsec:tetrad} 
Our core idea is to combine the ``tetrad differencing'' technique with the ``bounding-by-$c$'' technique together to derive bounds free of potentially complicated endogenous variables and unobserved fixed effects.

Consider a tetrad of distinct agents $ijhk\equiv(i,j,h,k)$. We first explain the tetrad-differencing technique. Specifically, consider the following tetrad event
\begin{equation}\label{eq:tetrad_event}
  Y_{ij}=Y_{hk}=1\ \text{and}\  Y_{ik}=Y_{jh}=0.
\end{equation}
which by \eqref{eq:link_shorthand} is equivalent to the event
\[v_{ij}\leq\delta_{ij},\quad v_{hk}\leq\delta_{hk},\quad v_{ik}>\delta_{ik},\quad v_{jh}>\delta_{jh}.
\]
which implies the following
\begin{equation}\label{eq:tetrad_sum}
  \Delta v:=v_{ij}+v_{hk}-v_{ik}-v_{jh}<\delta_{ij}+\delta_{hk}-\delta_{ik}-\delta_{jh}=:\Delta\delta.
\end{equation}
We observe that the fixed effects are all differenced out in $\Delta v$:
\begin{align}
  \Delta v & =v_{ij}+v_{hk}-v_{ik}-v_{jh}\nonumber\\
  & =(A_{i}+A_{j}+\varepsilon_{ij})+(A_{h}+A_{k}+\varepsilon_{hk})-(A_{i}+A_{k}+\varepsilon_{ik})-(A_{j}+A_{h}+\varepsilon_{jh})\nonumber\\
  & =\varepsilon_{ij}+\varepsilon_{hk}-\varepsilon_{ik}-\varepsilon_{jh}\nonumber\\
  & =:\Delta\varepsilon.
  \label{eq:delta_v}
\end{align}
Since $\Delta v=\Delta\varepsilon$ by~\eqref{eq:delta_v}, the inequality~\eqref{eq:tetrad_sum} is equivalent to
\begin{equation}\label{eq:tetrad_sum_simplified}
  \Delta\varepsilon<\Delta\delta.
\end{equation}

Now, we combine the above with the ``bounding-by-$c$'' technique. Specifically, consider the intersection of the tetrad event \eqref{eq:tetrad_event} with the event $\{\Delta\delta\leq c\}$, i.e.,
\begin{equation}\label{eq:tetrad_event_c}
  Y_{ij}=Y_{hk}=1\ \text{ and }\  Y_{ik}=Y_{jh}=0\  \text{ and }\ \Delta\delta\leq c,
\end{equation}
which by \eqref{eq:tetrad_sum_simplified} implies
\[
  \Delta\varepsilon<\Delta\delta \text{ and }\ \Delta\delta\leq c.
\]
The above further implies
\begin{equation}\label{eq:bound_U_event}
  \Delta\varepsilon\leq c,
\end{equation}
an inequality on the exogenous $\Delta\varepsilon$ without fixed effects or endogeneity issues.

Let $F_{\Delta}$ denote the CDF of $\Delta\varepsilon$. Since the intersection event~\eqref{eq:tetrad_event_c} implies~\eqref{eq:bound_U_event}, taking conditional probabilities given $\zeta_{ijhk}=\zeta$ yields
\begin{equation}\label{eq:bound_L}
  \E\left[Y_{ij}Y_{hk}(1-Y_{ik})(1-Y_{jh})\ind\{\Delta\delta\leq c\}\,\middle|\,\zeta\right]\ \leq\ \Prob\left(\Delta\varepsilon\leq c\right)\equiv F_{\Delta}(c)
\end{equation}
by the independence of $\varepsilon$ and $Z$ in Assumption \ref{ass:exogeneity}.\footnote{Although $\Delta\delta$ involves the endogenous covariates $X_{ij}$ and hence is not $\sigma(Z)$-measurable, the inequality \eqref{eq:bound_L} is valid because the bound $\ind\{\text{tetrad event}\}\cdot\ind\{\Delta\delta\leq c\}\leq\ind\{\Delta\varepsilon\leq c\}$ holds \emph{pointwise} (i.e., for every realization of $(Y,X,Z,A,\varepsilon)$), and taking conditional expectations of both sides preserves the inequality. The right-hand side $\Prob(\Delta\varepsilon\leq c\mid \zeta_{ijhk}=\zeta)=F_\Delta(c)$ follows from $\varepsilon\perp Z$, since $\zeta_{ijhk}$ is a measurable function of $(Z_1,\ldots,Z_n)$.}

Similarly, we can obtain another bound by considering the ``flipped'' event
\[
  Y_{ij}=Y_{hk}=0\ \text{  and  }\  Y_{ik}=Y_{jh}=1\ \text{  and  } \ \Delta\delta>c,
\]
which implies $\Delta\varepsilon>\Delta\delta>c$ and hence
\[(1-Y_{ij})(1-Y_{hk})Y_{ik}Y_{jh}\ind\{\Delta\delta>c\}\leq\ind\{\Delta\varepsilon>c\},\]
yielding
\begin{equation}
  \E\left[(1-Y_{ij})(1-Y_{hk})Y_{ik}Y_{jh}\ind\{\Delta\delta> c\}\,\middle|\,\zeta\right]\ \leq\  1-F_{\Delta}(c),
\end{equation}
or equivalently,
\begin{equation}\label{eq:bound_U}
  F_{\Delta}(c) \leq 1- \E\left[(1-Y_{ij})(1-Y_{hk})Y_{ik}Y_{jh}\ind\{\Delta\delta> c\}\,\middle|\,\zeta\right].
\end{equation}

To state our main results, it is helpful to make explicit the dependence of the index difference on the model parameters and the tetrad $ijhk$. In the following we write the tetrad index difference
\begin{equation}
  \Delta_{ijhk}(\theta):=\delta_{ij}(\theta)+\delta_{hk}(\theta)-\delta_{ik}(\theta)-\delta_{jh}(\theta)
\end{equation}
and two parametrized conditional probabilities (expectations)
\begin{align}\label{eq:pL_pU}
  p_{L}(\zeta,c;\theta) &:=\E\left[Y_{ij}Y_{hk}(1-Y_{ik})(1-Y_{jh})\ind\{\Delta_{ijhk}(\theta)\leq c\}\,\middle|\,\zeta\right],\\
  p_{U}(\zeta,c;\theta) &:=1-\E\left[(1-Y_{ij})(1-Y_{hk})Y_{ik}Y_{jh}\ind\{\Delta_{ijhk}(\theta)> c\}\,\middle|\,\zeta\right].
\end{align}

We now exploit the above to build identifying restrictions, which have two variants depending whether researchers impose an additional parametric assumption on $\varepsilon$.

We first explain the case where no parametric assumption is made on $\varepsilon$ and $F_\Delta$ is left unknown. In this case, \eqref{eq:bound_L} and \eqref{eq:bound_U} are not identifying restrictions per se since they involve the unknown $F_\Delta(c)$. However, \eqref{eq:bound_L} and \eqref{eq:bound_U} together imply
\[
  p_{L}(\zeta,c;\theta_0)\ \leq F_\Delta(c)\  \leq p_{U}(\zeta,c;\theta_0),
\]
and then, since the middle term is constant across $\zeta$,
\[
  \sup_{\zeta}p_{L}(\zeta,c;\theta_0)\ \leq F_\Delta(c)\  \leq \inf_{\zeta}p_{U}(\zeta,c;\theta_0),
\]
which becomes an identifying restriction once we drop the unknown middle term  $F_\Delta(c)$. We summarize our result in the following theorem.

\begin{theorem}[Tetrad Restrictions with Nonparametric $F_\Delta$]\label{thm:tetrad}
  Under Assumptions~\ref{ass:random_sampling}--\ref{ass:tetrad_id}, the true parameter $\theta_{0}$ belongs to the identified set
  \begin{equation}
    \Theta^{\mathrm{tetrad}}_{I}:=\Big\{\theta:\ \sup_{\zeta}p_{L}(\zeta,c;\theta)\leq\inf_{\zeta}p_{U}(\zeta,c;\theta)\ \text{for all }c\in\R\Big\},
    \label{eq:tetrad_main}
  \end{equation}
  where the supremum and infimum are taken over $\zeta$ in the support of $\zeta_{ijhk}$.

  Equivalently, define
  \begin{equation}
    Q^{\mathrm{tetrad}}(\theta):=\sup_{c\in\R}\left\{\sup_{\zeta}p_{L}(\zeta,c;\theta)-\inf_{\zeta}p_{U}(\zeta,c;\theta)\right\}.
    \label{eq:tetrad_criterion}
  \end{equation}
  Then $\Theta^{\mathrm{tetrad}}_{I}=\{\theta:Q^{\mathrm{tetrad}}(\theta)\leq0\}$.
\end{theorem}

\begin{remark}[Outer Region]\label{rem:outer_region}
  The set $\Theta_I^{\mathrm{tetrad}}$ is an outer region for the true parameter: it contains~$\theta_0$ but may also contain parameters~$\theta$ for which no i.i.d.\ pairwise shock distribution generates the observed tetrad probabilities.  This arises because the construction eliminates $F_\Delta(c)$ by intersecting bounds across~$\zeta$, without enforcing the constraint that $F_\Delta$ must be the distribution of $\varepsilon_{ij}+\varepsilon_{hk}-\varepsilon_{ik}-\varepsilon_{jh}$ for some i.i.d.\ pairwise law.  Incorporating this convolution structure could tighten the identified set, but we do not pursue this here.  Under the parametric logistic specification (Theorem~\ref{thm:tetrad_parametric}), the distribution $F_\Delta$ is fully determined, so the parametric identified set does not suffer from this relaxation.
\end{remark}

Next, we consider the case where researchers impose an additional parametric assumption on $\varepsilon$, so that $F_\Delta(c;\theta_\varepsilon)$ is unknown up to a finite-dimensional parameter $\theta_\varepsilon$. Practically, if $\varepsilon$ is assumed to follow a distribution with location-scale parameters such as the logistic distribution \citep{graham2017econometric} and the normal distribution \citep{dzemski2019empirical}, then there is no need to explicitly introduce the parameter $\theta_\varepsilon$ given the scope of location and scale normalization in model \eqref{eq:link_formation}. That said, we keep the notation $\theta_\varepsilon$ for theoretical completeness.

With $F_\Delta(c;\theta_\varepsilon)$ known up to $\theta_\varepsilon$, the inequalities \eqref{eq:bound_L} and \eqref{eq:bound_U} become identifying restrictions themselves, which we summarize in the following theorem.

\begin{theorem}[Tetrad Restrictions with Parametric $F_{\Delta}$]\label{thm:tetrad_parametric}
  Suppose Assumptions~\ref{ass:random_sampling}--\ref{ass:tetrad_id} hold and
  that $F_{\Delta}$ belongs to a known parametric family indexed by
  $\theta_\varepsilon\in\Theta_\varepsilon$, with CDF $F_{\Delta}(\,\cdot\,;\theta_\varepsilon)$.  Write $\overline{\theta}:=(\theta,\theta_\varepsilon)$ for the stacked parameter and $\overline{\theta}_{0}:=(\theta_{0},\theta_{\varepsilon,0})$ for its true value. Then $\overline{\theta}_{0}$ belongs to the identified set
  \begin{equation}
    \overline{\Theta}_{I}^{\mathrm{tetrad}}:=\Big\{\overline{\theta}:\ \sup_{\zeta}p_{L}(\zeta,c;\theta)\leq F_{\Delta}(c;\theta_\varepsilon)\leq\inf_{\zeta}p_{U}(\zeta,c;\theta)\ \text{for all }c\in\R\Big\}.
    \label{eq:tetrad_parametric_main}
  \end{equation}
  Equivalently, define the criterion
  \begin{equation}
    Q^{\mathrm{tetrad}}(\overline{\theta}):= \sup_{c\in\R}\ \max\left\{\sup_{\zeta}p_{L}(\zeta,c;\theta)-F_{\Delta}(c;\theta_\varepsilon),\ F_{\Delta}(c;\theta_\varepsilon)-\inf_{\zeta}p_{U}(\zeta,c;\theta)\right\}.
    \label{eq:tetrad_parametric_criterion}
  \end{equation}
  Then  $\overline{\Theta}^{\mathrm{tetrad}}_{I}=\left\{\overline{\theta}:\ Q^{\mathrm{tetrad}}(\overline{\theta})\leq0\right\}$.
\end{theorem}

Several features of Theorems \ref{thm:tetrad} and \ref{thm:tetrad_parametric} are worth noting.
First, the tetrad construction eliminates all individual fixed effects by
purely algebraic differencing within a four-agent configuration, so
identification does not rely on sufficient-statistic arguments (as in
parametric logit) nor on restrictions such as $A_{i}\perp Z_{i}$.
Second, the ``bounding by $c$'' step addresses the endogeneity of equilibrium
network statistics by treating the endogenous index components as random
variables and conditioning on observable events that imply inequalities in
$\Delta\varepsilon$; in particular, we never need to model, estimate, or
simulate the conditional distribution of $X_{ij}$.
Third, because the restrictions are derived from monotonicity and equilibrium
feasibility alone, they remain valid without computing the equilibrium mapping
$g$ and are robust to equilibrium multiplicity and unknown equilibrium
selection.
We note that the $n$-dependence of the covariate construction $w_n$ (introduced in Section~\ref{sec:model}) plays no role in the identification arguments above: the bounding-by-$c$ restrictions hold for each fixed network size $n$ and are purely algebraic. The $n$-dependence becomes relevant only at the estimation stage, where one must verify that the conditional probabilities $p_L$ and $p_U$ can be consistently estimated; see Remark~\ref{rem:id_vs_est} for a detailed discussion.


\subsection{Triad-Based Restrictions}

\label{subsec:triad} 

While tetrad restrictions achieve complete elimination of fixed effects,
they are not the only source of identifying restrictions. We now develop
complementary restrictions based on triads---subnetwork configurations involving three agents. Triad-based restrictions do not fully eliminate all fixed effects, but they have two practical advantages. First, triads are combinatorially more abundant than tetrads ($\binom{n}{3}$ versus $\binom{n}{4}$ configurations), yielding more observations per network. Second, they exploit different aspects of the latent variable structure and can tighten the identified set when combined with tetrad restrictions.

In each triad configuration, the differencing construction eliminates some agents' fixed effects but retains others.  The agents whose fixed effects survive---i.e., those with nonzero incidence load $\sigma_i\neq 0$ in the notation of Section~\ref{subsec:cycle}---are the \emph{retained} agents.  The residual latent composite depends on agent heterogeneity only through the retained agents, and by the i.i.d.\ structure of $(Z_i,A_i)$ (Assumption~\ref{ass:random_sampling}), its conditional law given the retained agents' exogenous characteristics does not depend on the characteristics of the differenced-out agents.  This yields bounds that are sharper than unconditional versions, because they avoid averaging over the heterogeneity in $A_i\mid Z_i$ for the retained agents. We consider several variants below.

\subsubsection*{Three-Link Triad Restrictions}

Consider a triad $(i,j,k)$ and the configuration where agent $i$ is linked
to both $j$ and $k$, but $j$ and $k$ are not linked.
Given the event $\{Y_{ij}=Y_{ik}=1,\,Y_{jk}=0\}$, the link inequalities imply
\begin{equation}
  v_{ij}+v_{ik}-v_{jk}<\delta_{ij}+\delta_{ik}-\delta_{jk}.
  \label{eq:triad_three_link_basic_ineq}
\end{equation}
The left-hand side equals
\begin{equation}
  u_{i,jk}:=v_{ij}+v_{ik}-v_{jk}=2A_{i}+\varepsilon_{ij}+\varepsilon_{ik}-\varepsilon_{jk},
  \label{eq:u_three_link}
\end{equation}
so intersecting with $\{\delta_{ij}+\delta_{ik}-\delta_{jk}\le c\}$ yields the indicator
bound
\[
  Y_{ij}Y_{ik}(1-Y_{jk}) \ind\{\delta_{ij}+\delta_{ik}-\delta_{jk}\le c\}\le\ind\{u_{i,jk}\le c\}.
\]
Unlike the tetrad case, the fixed effect $A_i$ does not cancel: the incidence loads are $\sigma_i=+2$, $\sigma_j=\sigma_k=0$, so agent~$i$ is the sole \emph{retained} agent whose fixed effect survives. The residual $u_{i,jk}=2A_i+\varepsilon_{ij}+\varepsilon_{ik}-\varepsilon_{jk}$ depends on $(A_i,\varepsilon_{ij},\varepsilon_{ik},\varepsilon_{jk})$, which by Assumptions~\ref{ass:random_sampling}--\ref{ass:iid_errors} are jointly independent of $(Z_j,Z_k)$ conditional on~$Z_i$. Consequently, the conditional distribution of $u_{i,jk}$ given $Z_i=z_i$ does not depend on $(Z_j,Z_k)$. Let
\begin{equation}\label{eq:G3_conditional}
  G_3(c\mid z_i):=\Prob(u_{i,jk}\le c\mid Z_i=z_i)
\end{equation}
denote this conditional CDF. Taking expectations conditional on $(Z_i=z_i,Z_j=z_j,Z_k=z_k)$ and then taking suprema over
$(z_j,z_k)$ for each fixed $z_i$ yields the moment-inequality bounds stated below.

\begin{proposition}[Three-Link Triad Bounds]\label{prop:three_link}
  For any $c\in\R$ and any $z_i$ in the support of $Z_i$:
  \begin{equation}
    \sup_{z_{j},z_{k}}\E\left[Y_{ij}Y_{ik}(1-Y_{jk})\ind\{\delta_{ij}+\delta_{ik}-\delta_{jk}\leq c\}\,\middle|\,z_i,z_j,z_k\right]\leq G_3(c\mid z_i)\label{eq:three_link}
  \end{equation}
  \begin{equation}
    G_3(c\mid z_i)\leq1-\sup_{z_{j},z_{k}}\E\left[(1-Y_{ij})(1-Y_{ik})Y_{jk}\ind\{\delta_{ij}+\delta_{ik}-\delta_{jk}>c\}\,\middle|\,z_i,z_j,z_k\right]\label{eq:three_link_lower}
  \end{equation}
  where the suprema are over $(z_j,z_k)$ in the support of $(Z_j,Z_k)$. Combining the two inequalities eliminates the unknown $G_3(c\mid z_i)$ for each fixed $z_i$ and~$c$.
\end{proposition}

One may then incorporate the information in Proposition~\ref{prop:three_link} into the characterization of the identified set by arguments analogous to those in Theorems~\ref{thm:tetrad} and~\ref{thm:tetrad_parametric}.

\subsubsection*{Two-Link Triad Restrictions}

A complementary restriction compares two links originating from a common agent.
Specifically, fix $i$ and consider the event that $i$ links to $j$ but not to $k$. Given the observable event $\{Y_{ij}=1,\,Y_{ik}=0\}$, the link inequalities imply
\[
  v_{ij}-v_{ik}<\delta_{ij}-\delta_{ik}.
\]
Since $v_{ij}-v_{ik}=(A_{i}+A_{j}+\varepsilon_{ij})-(A_{i}+A_{k}+\varepsilon_{ik})$, the fixed effect
$A_i$ cancels (it has incidence load $\sigma_i=0$), leaving
\begin{equation}
  u_{i,j,k}:=v_{ij}-v_{ik}=A_{j}+\varepsilon_{ij}-A_{k}-\varepsilon_{ik}.
  \label{eq:u_two_link}
\end{equation}
Using the ``bounding-by-$c$'' technique and intersecting the above with $\{\delta_{ij}-\delta_{ik}\le c\}$, we obtain
\[
  Y_{ij}(1-Y_{ik}) \ind\{\delta_{ij}-\delta_{ik}\le c\}\le\ind\{u_{i,j,k}\le c\}.
\]
Here agents $j$ and $k$ are the retained agents ($\sigma_j=+1$, $\sigma_k=-1$), while agent~$i$ is differenced out ($\sigma_i=0$). The residual $u_{i,j,k}=A_j+\varepsilon_{ij}-A_k-\varepsilon_{ik}$ depends on $(A_j,A_k,\varepsilon_{ij},\varepsilon_{ik})$, which is independent of~$Z_i$ but may depend on $(Z_j,Z_k)$ through the correlation of $A_j$ with $Z_j$ and $A_k$ with $Z_k$. Let
\begin{equation}\label{eq:G2_conditional}
  G_2(c\mid z_j,z_k):=\Prob(u_{i,j,k}\le c\mid Z_j=z_j,Z_k=z_k)
\end{equation}
denote the conditional CDF, which does not depend on~$Z_i$.
Replicating the ``flipped'' argument, taking expectations conditional on $(Z_i=z_i,Z_j=z_j,Z_k=z_k)$, and then taking the supremum over
$z_i$ for each fixed $(z_j,z_k)$, we obtain the following.

\begin{proposition}[Two-Link Triad Bounds]\label{prop:two_link}
  For any $c\in\R$ and any $(z_j,z_k)$ in the support of $(Z_j,Z_k)$:
  \begin{equation}
    \sup_{z_{i}}\E\left[Y_{ij}(1-Y_{ik})\ind\{\delta_{ij}-\delta_{ik}\leq c\}\,\middle|\,z_i,z_j,z_k\right]\leq G_2(c\mid z_j,z_k)
    \label{eq:two_link}
  \end{equation}
  \begin{equation}
    G_2(c\mid z_j,z_k)\leq 1-\sup_{z_{i}}\E\left[(1-Y_{ij})Y_{ik}\ind\{\delta_{ij}-\delta_{ik}>c\}\,\middle|\,z_i,z_j,z_k\right].
    \label{eq:two_link_flipped}
  \end{equation}
  where the suprema are over $z_i$ in the support of $Z_i$. Combining the two inequalities eliminates the unknown $G_2(c\mid z_j,z_k)$ for each fixed $(z_j,z_k)$ and~$c$.
\end{proposition}

\begin{remark}[Unconditional Triad Bounds]\label{rem:unconditional_triad}
  Marginalizing the conditional CDFs $G_3(c\mid z_i)$ and $G_2(c\mid z_j,z_k)$ over the retained agents' characteristics yields unconditional CDFs
  \[
    \bar G_3(c):=\Prob(u_{i,jk}\le c),\qquad \bar G_2(c):=\Prob(u_{i,j,k}\le c),
  \]
  which do not depend on any conditioning variable. One may derive corresponding bounds in Propositions \ref{prop:three_link} and \ref{prop:two_link}. For example, in Proposition \ref{prop:three_link} we can integrate both sides of the lower bounds over $Z_i\sim F_{Z_i}$ (and taking supremum afterwards), which yields
  \[
    \sup_{z_j,z_k}\E\bigl[Y_{ij}Y_{ik}(1-Y_{jk})\ind\{\delta_{ij}+\delta_{ik}-\delta_{jk}\le c\}\,\big|\,z_j,z_k\bigr]\;\le\;\bar G_3(c).
  \]
  Analogously, one may derive the upper bound in Proposition~\ref{prop:three_link}, as well as the corresponding bounds based on $\bar G_2(c)$ in Proposition~\ref{prop:two_link} by conditioning on $Z_i=z_i$ and profiling over $z_i$. These unconditional bounds are generally less sharp, since they further average over the heterogeneity in $A_i\mid Z_i$ for the retained agents, but may be easier to implement in practice because the conditional expectations involve fewer conditioning variables.
\end{remark}

\begin{remark}[Sharpening via Conditional Profiling]\label{rem:conditional_profiling}
Note that the bounds in Propositions~\ref{prop:three_link} and~\ref{prop:two_link} condition on the exogenous characteristics of the \emph{retained} agents whose fixed effects survive the differencing construction,  and profile only over the characteristics of the \emph{differenced-out} agents. Formally, the residual $U_S=\sum_{i\in S_R}\sigma_i A_i+\sum_e\omega_e\varepsilon_e$ depends on agent heterogeneity only through the retained set $S_R=\{i:\sigma_i\neq 0\}$; by the i.i.d.\ structure of $(Z_i,A_i)$ (Assumption~\ref{ass:random_sampling}) and $\varepsilon\perp Z$ (Assumption~\ref{ass:exogeneity}), the conditional law of $U_S$ given $Z_{S_R}=z_R$ does not depend on $Z_{S_0}$, where $S_0=\{i:\sigma_i=0\}$.

  The two triad cases illustrate the general principle with complementary structures: the three-link triad retains only the central agent ($S_R=\{i\}$, $S_0=\{j,k\}$) and so conditions on~$z_i$; the two-link triad retains the endpoint agents ($S_R=\{j,k\}$, $S_0=\{i\}$) and conditions on $(z_j,z_k)$. When $A_i$ and $Z_i$ are uncorrelated, the conditional and unconditional CDFs coincide and the bounds reduce to their unconditional counterparts; when $A_i$ and $Z_i$ are correlated, the conditional bounds are strictly tighter. This conditional profiling structure extends to the weighted differencing restrictions of Section~\ref{subsec:weighted} and the general weighted cycle-based restrictions of Section~\ref{subsec:cycle}.

  The triad bounds thus combine two distinct strategies for handling fixed effects: agents in $S_0$ have their fixed effects \emph{differenced out} by the signed-weight construction, while retained agents in $S_R$ have their fixed effects \emph{integrated out} by conditioning on their exogenous characteristics $Z_{S_R}$. The latter step is conceptually analogous to the ``integrate-out'' approach to fixed effects in \cite{gao2026identification}, where individual effects are eliminated by conditioning on covariates and exploiting a time-homogeneity condition on the distribution of $A_i$ given $Z_i$. In the present cross-sectional setting, the role of time homogeneity is played by the i.i.d.\ structure of $(Z_i,A_i)$ across agents (Assumption~\ref{ass:random_sampling}): once one conditions on $Z_i=z_i$, the conditional distribution of $A_i$ is a well-defined, agent-invariant object, and the residual CDF $G_3(c\mid z_i)$ or $G_2(c\mid z_j,z_k)$ becomes free of nuisance heterogeneity. The tetrad restrictions of Section~\ref{subsec:tetrad}, in contrast, are purely ``difference-out'' ($S_R=\varnothing$).
\end{remark}


\subsection{Weighted Differencing} \label{subsec:weighted} 

The tetrad and triad restrictions developed above use equal weights
on links. A natural extension is to consider weighted combinations
of link indicators, which can generate additional identifying restrictions
by exploiting different linear combinations of the latent inequalities.
We illustrate the idea with a single example; the systematic development
of weighted and more general differencing schemes is presented in
Section~\ref{subsec:cycle}.

Within a tetrad $(i,j,k,\ell)$, consider the link configuration:
\begin{equation}
  Y_{ij}Y_{ik}(1-Y_{i\ell})\label{eq:weighted_config}
\end{equation}
The event $\{Y_{ij}=Y_{ik}=1,\,Y_{i\ell}=0\}$ implies
$v_{ij}\le \delta_{ij}$, $v_{ik}\le \delta_{ik}$, and $v_{i\ell}>\delta_{i\ell}$.
Summing these inequalities with weights $(+1,+1,-2)$, we obtain
\[
  v_{ij}+v_{ik}-2v_{i\ell}<\delta_{ij}+\delta_{ik}-2\delta_{i\ell}.
\]
Since $v_{ij}+v_{ik}-2v_{i\ell}=(A_{j}+\varepsilon_{ij})+(A_{k}+\varepsilon_{ik})-2(A_{\ell}+\varepsilon_{i\ell})$, by
intersecting with $\{\delta_{ij}+\delta_{ik}-2\delta_{i\ell}\le c\}$ we can derive the bound
\begin{equation}
  Y_{ij}Y_{ik}(1-Y_{i\ell}) \ind\{\delta_{ij}+\delta_{ik}-2\delta_{i\ell}\leq c\}\leq\ind\{(A_{j}+\varepsilon_{ij})+(A_{k}+\varepsilon_{ik})-2(A_{\ell}+\varepsilon_{i\ell})\leq c\}.
  \label{eq:weighted_bound}
\end{equation}
The incidence loads are $\sigma_i=0$, $\sigma_j=+1$, $\sigma_k=+1$, $\sigma_\ell=-2$, so agent~$i$ is differenced out while $j$, $k$, and $\ell$ are retained. The weighted latent index on the right-hand side depends on $(A_j,A_k,A_\ell,\varepsilon_{ij},\varepsilon_{ik},\varepsilon_{i\ell})$, which is independent of~$Z_i$ but may depend on $(Z_j,Z_k,Z_\ell)$ through the correlation of each $A_m$ with~$Z_m$. Let
\begin{equation}\label{eq:Gw_conditional}
  G_{\mathrm{w}}(c\mid z_j,z_k,z_\ell):=\Prob\left((A_{j}+\varepsilon_{ij})+(A_{k}+\varepsilon_{ik})-2(A_{\ell}+\varepsilon_{i\ell})\leq c\,\middle|\,Z_j=z_j,Z_k=z_k,Z_\ell=z_\ell\right)
\end{equation}
denote the conditional CDF, which does not depend on~$Z_i$. Taking expectations conditional on $(Z_i=z_i,Z_j=z_j,Z_k=z_k,Z_\ell=z_\ell)$ and then taking the supremum over~$z_i$ for each fixed $(z_j,z_k,z_\ell)$ yields the following.

\begin{proposition}[Weighted Differencing Bound]\label{prop:weighted}
  Within a tetrad $(i,j,k,\ell)$, for any $c\in\R$ and any $(z_j,z_k,z_\ell)$ in the support of $(Z_j,Z_k,Z_\ell)$:
  \begin{equation}
    \sup_{z_{i}}\E\left[Y_{ij}Y_{ik}(1-Y_{i\ell}) \ind\{\delta_{ij}+\delta_{ik}-2\delta_{i\ell}\leq c\}\,\middle|\,z_i,z_j,z_k,z_\ell\right]\leq G_{\mathrm{w}}(c\mid z_j,z_k,z_\ell).
    \label{eq:weighted_expectation}
  \end{equation}
  \begin{equation}
    G_{\mathrm{w}}(c\mid z_j,z_k,z_\ell)\leq 1-\sup_{z_{i}}\E\left[(1-Y_{ij})(1-Y_{ik})Y_{i\ell}\ind\{\delta_{ij}+\delta_{ik}-2\delta_{i\ell}>c\}\,\middle|\,z_i,z_j,z_k,z_\ell\right].
    \label{eq:weighted_expectation_flipped}
  \end{equation}
  where the suprema are over $z_i$ in the support of $Z_i$. Combining the two inequalities eliminates the unknown $G_{\mathrm{w}}(c\mid z_j,z_k,z_\ell)$ for each fixed $(z_j,z_k,z_\ell)$ and~$c$.
\end{proposition}

The potential advantage of weighted differencing is that it generates
a richer family of moment inequalities, indexed by the weights,
which may tighten the identified set beyond what is achievable with
the equal-weight tetrad and triad restrictions alone.

\subsection{General Weighted Cycle-Based Differencing}

\label{subsec:cycle} 

It should now be
clear that our core identification idea can be generalized
further beyond the identifying restrictions developed in Sections~\ref{subsec:tetrad}--\ref{subsec:weighted}, which
are all instances of a unified principle: assign signed weights to
links in a subgraph, and when the weighted sum of link indicators
factors into a product of indicators, the ``bounding by $c$'' technique
yields bounds involving a weighted sum of latent variables. Fixed
effects cancel at any agent whenever the sum of weights on links incident
to that agent equals zero. We now formalize this observation.

\begin{definition}[Weighted Link Configuration]\label{def:signed_config}
  Let $S=\{i_{1},\ldots,i_{m}\}\subseteq\{1,\ldots,n\}$ be a subset of
  distinct agents. A \emph{weighted link configuration} on $S$ consists
  of a (finite) set of links $E_S$ together with a weight function
  $\omega:E_S\to\mathbb{Z}$ (or $\mathbb{R}$), where $\omega_{e}>0$ indicates a
  ``present'' link and $\omega_{e}<0$ indicates an ``absent'' link.
  The \emph{weighted incidence sum} at agent $i\in S$ is
  \[
    \sigma_{i}:=\sum_{e\in E_S:i\in e}\omega_{e}.
  \]
  where $i \in e$ means that $i$ is one of the two agents involved in link $e$.
\end{definition}

Clearly, a weighted link configuration achieves \emph{complete fixed-effect elimination}
if $\sigma_{i}=0$ for every $i\in S$, i.e., if the weights incident to each
agent sum to zero. The identifying restrictions from earlier sections correspond to specific weighted link configurations:

\begin{enumerate}
  \item \textbf{Tetrad (Section~\ref{subsec:tetrad}):} The 4-cycle $\{ij,jh,hk,ki\}$
    with alternating signs $E^{+}=\{ij,hk\}$, $E^{-}=\{ik,jh\}$. Every
    agent has signed incidence sum $\sigma=0$: complete fixed-effect elimination.
  \item \textbf{Three-link triad (Section~\ref{subsec:triad}):} The triangle
    $\{ij,ik,jk\}$ with $E^{+}=\{ij,ik\}$, $E^{-}=\{jk\}$. Node $i$
    has $\sigma_{i}=+2$; nodes $j,k$ have $\sigma_{j}=\sigma_{k}=0$.
    Partial elimination: $A_{j}$ and $A_{k}$ cancel but $2A_{i}$ remains.
  \item \textbf{Weighted differencing (Section~\ref{subsec:weighted}):}
    The star $\{ij,ik,i\ell\}$ with weights $+1,+1,-2$, which can
    be decomposed as $E^{+}=\{ij,ik\}$ and $E^{-}=\{i\ell\}$. Node $i$ has $\sigma_{i}=0$; the other nodes have
    nonzero incidence. Partial elimination: $A_{i}$ cancels but $A_{j}$, $A_{k}$, and $A_{\ell}$
    remain.
  \item \textbf{Hexad (New Example):} Consider the 6-cycle
    \[
      \{i_{1}i_{2}, i_{2}i_{3}, i_{3}i_{4}, i_{4}i_{5}, i_{5}i_{6}, i_{6}i_{1}\},
    \]
    with alternating signs $E^{+}=\{i_{1}i_{2}, i_{3}i_{4}, i_{5}i_{6}\}$ and
    $E^{-}=\{i_{2}i_{3}, i_{4}i_{5}, i_{6}i_{1}\}$. Every agent has $\sigma=0$:
    complete fixed-effect elimination using six agents.
\end{enumerate}

We are now ready to present an umbrella result for general weighted link
configurations. Fix a weighted link configuration $(E_S,\omega)$ on $S$, and let
\begin{equation}
  E^{+}:=\{e\in E_S:\omega_{e}>0\},\qquad E^{-}:=\{e\in E_S:\omega_{e}<0\}.
  \label{eq:Eplus_Eminus}
\end{equation}
Given the event $\{\prod_{e\in E^{+}}Y_{e}\prod_{e\in E^{-}}(1-Y_{e})=1\}$,
each $e\in E^{+}$ implies $v_{e}\le\delta_{e}$ and each $e\in E^{-}$ implies
$v_{e}>\delta_{e}$. Multiplying the corresponding inequalities by $\omega_{e}$ and
summing yields
\[
  \sum_{e\in E_S}\omega_{e}v_{e}<\sum_{e\in E_S}\omega_{e}\delta_{e}.
\]
Then
\[
  \sum_{e\in E_S}\omega_{e}v_{e}=\sum_{i\in S}\sigma_{i}A_{i}+\sum_{e\in E_S}\omega_{e}\varepsilon_{e}=:U_S.
\]
Hence, the fixed effect of any agent with $\sigma_i=0$ is differenced out, while agents with
$\sigma_i\neq 0$ contribute residual fixed-effect terms. As in the triad and weighted-differencing cases, we can sharpen the resulting bounds by conditioning on the exogenous characteristics of the retained agents. Since $U_S=\sum_{i\in S_R}\sigma_i A_i+\sum_{e\in E_S}\omega_e\varepsilon_e$, where $S_R:=\{i\in S:\sigma_i\neq 0\}$, the conditional law of $U_S$ given $Z_{S_R}=z_R$ does not depend on $Z_{S_0}$: the retained agents' fixed effects are integrated out through the conditional distribution of $A_i\mid Z_i=z_i$ for $i\in S_R$, while the differenced-out agents' characteristics serve as profiling variables. Adding the ``bounding-by-$c$'' event
$\{\sum_{e\in E_S}\omega_{e}\delta_{e}\leq c\}$ yields a bound in terms of $U_S$, which translates to the following.

\begin{proposition}[General Weighted Cycle-Based Identifying Restrictions]\label{prop:cycle}
  Let $S$ be a subset of distinct agents and
  $(E_S,\omega)$ be a weighted link configuration on $S$. Define the subsets
  \begin{equation}
    S_{0}:=\{i\in S:\ \sigma_{i}=0\},\qquad S_R:=\{i\in S:\ \sigma_i\neq 0\}=S\setminus S_0.
    \label{eq:S0_def}
  \end{equation}
  Then for any $c\in\R$, any realization $z_{S_R}$ of $Z_{S_R}$, and any realization $z_{S_0}$ of $Z_{S_0}$:
  \begin{equation}
    \E\left[\prod_{e\in E^{+}}Y_{e}\prod_{e\in E^{-}}(1-Y_{e})\ind\left\{\sum_{e\in E_S}\omega_{e}\delta_{e}\leq c\right\}\,\middle|\,Z_{S}=z_{S}\right]\leq F_{U_{S}}(c\mid z_{S_R}),
    \label{eq:cycle_bound}
  \end{equation}
  and
  \begin{equation}
    \E\left[\prod_{e\in E^{+}}(1-Y_{e})\prod_{e\in E^{-}}Y_{e}\ind\left\{\sum_{e\in E_S}\omega_{e}\delta_{e}> c\right\}\,\middle|\,Z_{S}=z_{S}\right]\leq 1-F_{U_{S}}(c\mid z_{S_R}),
    \label{eq:cycle_bound_flipped}
  \end{equation}
  where $F_{U_{S}}(c\mid z_{S_R}):=\Prob(U_S\leq c\mid Z_{S_R}=z_{S_R})$ does not depend on $z_{S_{0}}$.
  Taking suprema over $z_{S_0}$ for each fixed $z_{S_R}$ and combining the upper and lower bounds eliminates the unknown conditional CDF $F_{U_S}(\cdot\mid z_{S_R})$.

  When $S_R=\varnothing$ (complete fixed-effect elimination, as in the tetrad case), the conditional CDF $F_{U_S}(c\mid z_{S_R})$ reduces to the unconditional CDF $F_{U_S}(c)$, and the bounds recover the purely ``difference-out'' restrictions of Section~\ref{subsec:tetrad}.

\end{proposition} %

To apply the cycle-based restrictions in Proposition~\ref{prop:cycle} and aggregate them over a class $\mathcal{W}$ that involve subnetwork structures larger than tetrads, we need an analogue of Assumption~\ref{ass:tetrad_id} that guarantees identification of the required subnetwork conditional probabilities.

\begingroup
\renewcommand{\theassumption}{4$'$}
\begin{assumption}[Consistent Estimation of Subnetwork Conditional Probabilities]\label{ass:W_id}
  Fix a class $\mathcal{W}$ of admissible weighted link configurations $(S,E_S,\omega)$.
  For any $(S,E_S,\omega)\in\mathcal{W}$, let $\mathcal{E}_{S}$ be an event measurable with respect to the observable subnetwork $\sigma$-algebra
  \[
    \sigma_{S}:=\sigma\big(Y_{ij},X_{ij},Z_{i}: i,j\in S\big).
  \]
  Suppose that the conditional probability $\Prob(\mathcal{E}_{S}\mid Z_{S}=z_{S})$ can be consistently estimated from observable data in a single large network.
\end{assumption}
\endgroup

\begin{theorem}[Identified Set via Aggregation over Admissible Weighted Link Configurations]\label{thm:agg_identset}
  Fix a class $\mathcal{W}$ of admissible weighted link configurations $(S,E_S,\omega)$.
  For each $(S,E_S,\omega)\in\mathcal{W}$, let $S_{0}$ and $S_R$ be as in \eqref{eq:S0_def}, and define, for any $c\in\R$,
  \begin{align*}
    p^{(S,\omega)}_{L}(z_{S},c;\theta)
    &:=\E\left[\prod_{e\in E^{+}}Y_{e}\prod_{e\in E^{-}}(1-Y_{e})\ind\left\{\sum_{e\in E_S}\omega_{e}\delta_{e}(\theta)\leq c\right\}\,\middle|\,Z_{S}=z_{S}\right],\\
    p^{(S,\omega)}_{U}(z_{S},c;\theta)
    &:=1-\E\left[\prod_{e\in E^{+}}(1-Y_{e})\prod_{e\in E^{-}}Y_{e}\ind\left\{\sum_{e\in E_S}\omega_{e}\delta_{e}(\theta)> c\right\}\,\middle|\,Z_{S}=z_{S}\right].
  \end{align*}
  Then, under Assumptions~\ref{ass:random_sampling}--\ref{ass:iid_errors} and \ref{ass:W_id}, the true parameter $\theta_{0}$ belongs to the identified set
  \begin{equation}
    \Theta^{\mathcal{W}}_{I}:=\Big\{\theta:\ \sup_{(S,E,\omega)\in\mathcal{W}}\ \sup_{z_{S_R}}\ \sup_{c\in\R}\ \big[\sup_{z_{S_{0}}}p^{(S,\omega)}_{L}(z_{S},c;\theta)-\inf_{z_{S_{0}}}p^{(S,\omega)}_{U}(z_{S},c;\theta)\big]\leq0\Big\}.
    \label{eq:identified_set_aggregation}
  \end{equation}
\end{theorem}

\begin{remark}[More Restrictions versus Stronger Assumption \ref{ass:W_id}]\label{rem:longer_cycles_tradeoff}
  Allowing longer cycles and, more generally, larger subnetworks in $\mathcal{W}$ can tighten the identified set by adding additional conditional moment inequalities. However, the cost is that identification and estimation of the resulting subnetwork conditional probabilities requires a correspondingly stronger version of Assumption \ref{ass:W_id}. In finite sample, involving longer cycles can be more data-demanding because larger configurations are rarer, and furthermore the estimation can be less well-behaved since larger subnetworks suffer from more salient network dependence issues. As noted in Remark~\ref{rem:unconditional_triad}, one may also obtain coarser but easier-to-implement bounds by integrating $F_{U_S}(c\mid z_{S_R})$ over some or all of $Z_{S_R}$, thereby reducing the number of conditioning variables at the cost of sharpness.
\end{remark}

\begin{remark}[Why Differencing Is Necessary]\label{rem:no_pure_integrate}
  One might ask whether purely ``integrating out'' fixed effects (without any differencing) can generate identifying restrictions.  Consider a single link $(i,j)$ with weight $\omega_{ij}=+1$, so that $S_R=\{i,j\}$ and $S_0=\varnothing$.  The bounding-by-$c$ technique yields
  \[
    \E\bigl[Y_{ij}\ind\{\delta_{ij}\le c\}\,\big|\,Z_i=z_i,Z_j=z_j\bigr]\;\le\; F_{U_S}(c\mid z_i,z_j),
  \]
  where $F_{U_S}(c\mid z_i,z_j)=\Prob(A_i+A_j+\varepsilon_{ij}\le c\mid Z_i=z_i,Z_j=z_j)$ is unknown. With $S_0=\varnothing$, there is no profiling variable over which the right-hand side remains constant, so the unknown nuisance CDF cannot be eliminated by the sandwich argument.  In other words, the ``integrate-out'' step can sharpen bounds that differencing creates (by conditioning on $Z_{S_R}$ rather than marginalizing over it; see Remark~\ref{rem:unconditional_triad}), but it cannot produce restrictions on its own.  This parallels the panel-data setting of \citet{gao2026identification}, where the ``integrate-out'' approach likewise requires comparing observations across time periods, which is an implicit differencing step to eliminate the nuisance distribution.
\end{remark}


\section{Primitive Conditions for Assumption~\ref{ass:tetrad_id}}
\label{sec:primitive}

Define the joint surplus from a link
between $i$ and $j$, viewed as a function of the endogenous covariate
value $x$, as
\begin{equation}\label{eq:surplus_fn}
  V_{ij}(x) \;:=\; Z_{ij}'\beta_0 + x'\gamma_0 - A_i - A_j - \varepsilon_{ij}.
\end{equation}
Because the endogenous covariate $X_{ij} = \phi_{ij}(Y,Z)$ depends on
the realized network, the value of $x$ at which the surplus is
evaluated is itself an equilibrium object.  Each potential link
$(i,j)$ falls into exactly one of three categories, determined
entirely by the \emph{exogenous} primitives $(Z,A,\varepsilon)$:

\begin{enumerate}[label=(\roman*)]
  \item \textbf{Robustly present:}
    $\inf_x V_{ij}(x) > 0$\,---\,the surplus is positive for every
    possible value of the endogenous covariate, so the link forms in
    every equilibrium.
  \item \textbf{Robustly absent:}
    $\sup_x V_{ij}(x) \leq 0$\,---\,the surplus is non-positive for
    every possible value, so the link is absent in every equilibrium.
  \item \textbf{Non-robust:}
    $\sup_x V_{ij}(x) > 0$ and $\inf_x V_{ij}(x) \leq 0$\,---\,the
    link status depends on the equilibrium value of $X_{ij}$ and hence
    on the rest of the network.
\end{enumerate}

\begin{definition}[Strategic Neighborhood]\label{def:strat_nbhd-LM}
  The \emph{non-robustness indicator} for the potential link between
  $i$ and $j$ is
  \[
    D_{ij}
    \;:=\;
    \ind\!\bigl\{\sup_x V_{ij}(x)>0\bigr\}
    \cdot
    \ind\!\bigl\{\inf_x V_{ij}(x)\le 0\bigr\},
  \]
  where the supremum and infimum are over
  $x\in[-\bar x,\bar x]^{d_x}$.
  Let $C_i$ denote $i$'s connected component in the graph with
  adjacency matrix $\bm{D}=(D_{ij})$ (the \emph{non-robustness
  graph}), and let $\bm\Pi$ be the graph with adjacency matrix
  $\Pi_{ij}=\ind\{\inf_x V_{ij}(x)>0\}$ (the \emph{robust-link
  graph}).  The \emph{strategic neighborhood} of agent~$i$ is the set
  \[
    C_i^+
    \;:=\;
    \bigcup_{j\in C_i}\N_{\bm\Pi}(j,1),
  \]
  i.e., $C_i$ together with all agents who are robust neighbors of
  some member of~$C_i$.
\end{definition}

The connected components $C_i$ of the non-robustness graph~$\bm{D}$ partition the agent set $\mathcal{N}_n$, but the strategic neighborhoods $C_i^+$ can overlap: an agent who is a robust neighbor of two distinct non-robustness components belongs to both of their strategic neighborhoods.  Nevertheless, strategic neighborhoods localize the equilibrium: the link outcomes within any strategic neighborhood $C^+$ depend only on the primitives of agents in~$C^+$.  To see why, consider a non-robust link $(a,b)$ with $a,b\in C$ for some component~$C$.  The endogenous covariate $X_{ab}=\phi_{ab}(Y,Z)$ depends on other link outcomes, but only links incident to $a$ or $b$ contribute.  Any such link $(a,k)$ is either robustly absent ($Y_{ak}=0$, no contribution), robustly present ($k\in\N_{\bm\Pi}(a,1)\subset C^+$, value determined by primitives of $\{a,k\}\subset C^+$), or non-robust ($k\in C_a=C$, status determined by the equilibrium within~$C^+$).  Hence the equilibrium on~$C^+$ forms a self-contained fixed-point system in the primitives of agents in~$C^+$ \citep[Theorem~1]{leung2019treatment}.  We formalize this via the following assumptions.

\begin{assumption}[Local Externalities]\label{ass:local_ext}
  The endogenous covariate function $\phi_{ij}$ in~\eqref{eq:endogenous_X} depends only on links incident to $i$ or~$j$: that is, $X_{ij}=\phi_{ij}(Y,Z)$ is measurable with respect to $\sigma\bigl(\{Y_{ik},Y_{jk}:k\in\mathcal{N}_n\},Z\bigr)$.
\end{assumption}

Assumption~\ref{ass:local_ext} is satisfied by common friends $\text{CF}_{ij}=\sum_k Y_{ik}Y_{jk}$, the Jaccard index, degree statistics, and other standard endogenous covariates used in the network formation literature.

\begin{assumption}[Bounded Endogenous Covariates]\label{ass:bdd_x}
  The endogenous covariate takes values in a known compact set:
  $X_{ij}\in[-\bar x,\bar x]^{d_x}$ almost surely for some
  constant~$\bar x<\infty$ that does not depend on~$n$.
\end{assumption}

Unscaled common-friends counts $\mathrm{CF}_{ij}=\sum_k Y_{ik}Y_{jk}$ can grow with network size in dense networks. In the sparse regime of Assumption~\ref{ass:bdd_nbhd} below, however, expected degree is $O(1)$, so the expected number of common friends is also $O(1)$. The Jaccard index takes values in $[0,1]$ by construction, and normalized common friends $\mathrm{CF}_{ij}/n$ are bounded as well. For unnormalized count statistics in denser networks, explicit trimming or normalization would be needed to satisfy Assumption~\ref{ass:bdd_x}.

\begin{assumption}[Decentralized Selection]\label{ass:decentral-lm}
  For any $n\in\mathbb{N}$ and any strategic neighborhood~$C^+$, the
  equilibrium selection mechanism satisfies
  $Y_{C^+}=\eta_{|C^+|}(\tilde Z_{C^+},\varepsilon_{C^+})$, where
  $\eta_{|C^+|}(\cdot)$ depends only on the augmented types and shocks
  of agents in~$C^+$.
\end{assumption}

Assumption~\ref{ass:decentral-lm} requires that equilibrium play within
each strategic neighborhood is determined locally.  As discussed in
\cite{leung2019treatment}, this is satisfied by myopic best-response
dynamics, which are widely used in the theoretical and econometric
literature on network formation.

\begin{assumption}[Bounded Strategic Neighborhoods]%
  \label{ass:bdd_nbhd}
  The strategic neighborhoods satisfy
  $\mathbb{E}[|C_i^+|]=O(1)$ as $n\to\infty$.
\end{assumption}

Assumption~\ref{ass:bdd_nbhd} ensures that chains of strategic
dependencies do not propagate through the network.  It follows from
the branching-process domination argument in Theorem~1 of
\cite{leung2019treatment} under sparsity and subcriticality.
Intuitively, exploring $C_i^+$ via breadth-first search is akin to
growing a subcritical branching process that almost surely terminates.

\begin{proposition}[Packing Argument for
  Assumption~\ref{ass:tetrad_id}]%
  \label{prop:packing}
  Suppose Assumptions~\ref{ass:random_sampling}--\ref{ass:iid_errors},
  \ref{ass:local_ext}--\ref{ass:bdd_nbhd} hold, and that
  $Z_{i}$ has finite support.
  Then Assumption~\ref{ass:tetrad_id} is satisfied: for every
  $z_S=(z_i,z_j,z_h,z_k)$ in the support of~$Z_S$ and every tetrad event
  $E_{ijhk}\in\sigma_{ijhk}$, the conditional probability
  $\mathbb{P}(E_{ijhk}\mid Z_S=z_S)$ can be consistently
  estimated from a single large network.
\end{proposition}

The proof is given in Appendix~\ref{app:proofs}.

\begin{remark}[From Pointwise to Uniform Convergence]%
  \label{rem:uniform_convergence}
  Proposition~\ref{prop:packing} establishes pointwise consistency of the conditional probability estimator for fixed conditioning values.  The operational criterion $Q^{\mathrm{tetrad}}(\theta)$ involves suprema over $\zeta$ and $c$, and evaluation at varying~$\theta$, which requires uniform convergence guarantees not established here.  Closing this gap---for instance via empirical process theory for network-dependent data, or by reducing the continuous $c$-supremum to a finite grid using monotonicity of~$F_\Delta$---is an important direction for the formal econometric theory of inference in this setting.  For the present paper, we focus on the population-level identification results and treat the simulation evidence as illustrative.
\end{remark}

\begin{remark}[Continuous Individual Characteristics]%
  \label{rem:continuous_Z}
  The finite-support condition in Proposition~\ref{prop:packing} is
  imposed on the individual-level characteristics $Z_i$.  When $Z_i=(\Xi_i,W_i)$ is continuously distributed, this can be accommodated by discretizing the individual-level rescaled position $\bar\Xi_i$ into finitely many spatial bins, producing a discrete individual-level type $\bar Z_i^{\mathrm{disc}}$. The cell estimator then conditions on exact matches in $(\bar Z_i^{\mathrm{disc}},\bar Z_j^{\mathrm{disc}},\bar Z_h^{\mathrm{disc}},\bar Z_k^{\mathrm{disc}})$. Any discrete components of $(W_i,W_j)$ are retained as-is, yielding a finite-support individual characteristic and hence a finite-support dyadic covariate $Z_{ij}$.

  At the population level, the conditional probability
  $p(z_S)=\mathbb{P}(E_{ijhk}\mid Z_S=z_S)$ is
  well-defined as a regular conditional expectation at every~$z_S$
  at which~$Z_S$ has positive density; finite support is not
  required for the object itself to exist.  The identifying
  restrictions of Theorems~\ref{thm:tetrad}
  and~\ref{thm:tetrad_parametric} likewise hold pointwise for each
  such~$z_S$ (equivalently, for each~$\zeta$ in the support of~$\zeta_{ijhk}$), regardless of whether the support is finite or
  continuous.

  Finite support enters only at the estimation stage: the cell
  estimator and the packing
  argument both require exact matches $\{Z_{S,t}=z_S\}$, which have
  positive probability only when~$Z_{i}$ is discrete.  Such
  discretization is without loss for validity---the identifying
  restrictions hold for every type-quadruplet---though it is conservative: coarser
  bins pool heterogeneous covariate values, yielding fewer distinct
  conditioning values and hence fewer moment inequality restrictions,
  which may widen the identified set.  Finer discretization recovers
  more identifying power, approaching the continuous-$Z_i$ identified
  set in the limit.
\end{remark}


\section{Point Identification}\label{sec:point_id}


The identified sets in Theorems~\ref{thm:tetrad} and~\ref{thm:tetrad_parametric}
are defined by conditional moment inequalities and are
therefore generally set-valued.  It is natural to ask when these
restrictions are sharp enough to yield point identification.  In this
section we give conditions under which
$\theta_0=(\beta_0,\gamma_0)$ is point identified,
exploiting the algebraic structure of the logistic distribution.

When $\gamma_0\neq 0$, two complications arise.  First, the
endogenous covariates $X_{ij}=\phi_{ij}(Y,Z)$ create dependence among
the four tetrad link outcomes, so the product formula for
$\Prob(\text{tetrad}\mid\zeta,A)$ breaks down.  Second, the tetrad
index difference
$\Delta\delta=\Delta Z'\beta_0+\Delta X'\gamma_0$ depends on
$A$ through the equilibrium, so the conditional-on-$A$ odds ratio
varies with~$A$ and cannot be factored out when integrating
over~$A$.  We handle both problems by
\emph{conditioning on the realized endogenous covariate values} and
restricting attention to \emph{isolated tetrads} in which
the endogenous covariates for the four tetrad links are invariant
across the tetrad and flipped patterns.

\subsection{Setup and Definitions}

Fix a tetrad of distinct agents $(i,j,h,k)$.  Recall that the four
\emph{tetrad links} are
$E_t:=\{ij,hk,ik,jh\}$; the remaining two links within the
quadruplet are the \emph{diagonal links} $\{ih,jk\}$.  Write
$Y_{-E_t}$ for the collection of all link outcomes outside~$E_t$ and
$\varepsilon_{E_t}:=(\varepsilon_{ij},\varepsilon_{hk},\varepsilon_{ik},\varepsilon_{jh})$ for the tetrad shocks.
Define the \emph{tetrad endogenous covariate vector}
\begin{equation}\label{eq:Xt}
  X_t:=(X_{ij},X_{hk},X_{ik},X_{jh})
\end{equation}
and the \emph{tetrad exogenous covariate vector} $\zeta_{ijhk}$ as in~\eqref{eq:zeta_ijhk_app}.

\subsection{Tetrad Isolation}

We introduce three structural conditions on admissible tetrads.
Together they guarantee that the log-odds of the tetrad pattern
versus the flipped pattern is a linear function of the covariates,
allowing the fixed effects to cancel.

For an admissible tetrad $(i,j,h,k)$ with $Y_{ih}=Y_{jk}=0$,
define the \emph{tetrad pattern} and \emph{flipped pattern}
\begin{align}
  T_t &:=\{Y_{ij}=Y_{hk}=1,\;Y_{ik}=Y_{jh}=0\},\label{eq:Tt}\\
  F_t &:=\{Y_{ij}=Y_{hk}=0,\;Y_{ik}=Y_{jh}=1\}.\label{eq:Ft}
\end{align}

\begin{assumption}[Comparison-Pattern Invariance]\label{ass:cpi}
  For any admissible tetrad $(i,j,h,k)\in\mathcal{T}_n$ and each
  tetrad link $e\in E_t$, the endogenous covariate
  $X_e=\phi_e(Y_{-e},Z)$ takes the same value under the tetrad
  pattern~$T_t$ and the flipped pattern~$F_t$, holding all non-tetrad
  links fixed.  Formally, there exists a function $\tilde\phi_e$ such
  that
  \begin{equation}\label{eq:cpi}
    X_e\big|_{T_t}
    \;=\;
    X_e\big|_{F_t}
    \;=\;
    \tilde\phi_e(Y_{-E_t},\,Z).
  \end{equation}
\end{assumption}

Assumption~\ref{ass:cpi} requires that the endogenous covariates of
the four tetrad links are invariant across the two comparison
patterns.  It is weaker than requiring that
$X_e$ not depend on any tetrad link outcome globally (as would follow
from the bilinear condition $X_{ij}=\psi(\{Y_{il}Y_{jl}\}_{l\neq
i,j},Z)$): bilinearity implies~\eqref{eq:cpi}
globally, whereas Assumption~\ref{ass:cpi} only requires it on the
comparison sample $T_t\cup F_t$.  The following examples verify the
condition for the two leading specifications.

\begin{example}[Common Friends]\label{ex:cf_cpi}
  Under $X_{ij}=\sum_{l\neq i,j}Y_{il}Y_{jl}$, each within-quadruplet
  term $Y_{al}Y_{bl}$ with $l\in\{i,j,h,k\}\setminus\{a,b\}$ involves
  a diagonal link ($Y_{ih}=0$ or $Y_{jk}=0$) and therefore vanishes.
  Hence $X_e|_{T_t\cup F_t}=\sum_{l\notin\{i,j,h,k\}}Y_{al}Y_{bl}$,
  which depends only on $Y_{-E_t}$.
\end{example}

\begin{example}[Standard Jaccard Index]\label{ex:jac_cpi}
  Under $X_{ij}=|N(i)\cap N(j)|/|N(i)\cup N(j)|$ with
  $N(i)=\{l:Y_{il}=1\}$, define the \emph{external} quantities
  \[
    \bar{c}_{ij}:=\sum_{l\notin\{i,j,h,k\}}Y_{il}Y_{jl},
    \qquad
    \bar{d}_i:=\sum_{l\notin\{j,h,k\}}Y_{il}.
  \]
  On $T_t\cup F_t$, the within-quadruplet common-friend contributions
  vanish (same argument as Example~\ref{ex:cf_cpi}), and each
  endpoint has exactly one within-tetrad neighbor under both $T_t$ and
  $F_t$ (e.g., agent~$i$ links to exactly one of $\{j,k\}$ in both
  patterns), giving
  \[
    X_{ij}\big|_{T_t\cup F_t}
    =\frac{\bar{c}_{ij}}
    {\bar{d}_i+\bar{d}_j+2-\bar{c}_{ij}},
  \]
  and similarly for $X_{hk},X_{ik},X_{jh}$.
  All quantities on the right depend only on $(Y_{-E_t},Z)$.
\end{example}

\begin{assumption}[Tetrad Exogeneity]\label{ass:tetrad_exog}
  For tetrads in the admissible set $\mathcal{T}_n$ (defined below),
  the non-tetrad link outcomes $Y_{-E_t}$ are measurable with respect
  to $\sigma(Z,A,\varepsilon_{-E_t})$, where
  $\varepsilon_{-E_t}:=\{\varepsilon_{ab}:(a,b)\notin E_t\}$.
\end{assumption}

Assumption~\ref{ass:tetrad_exog} requires that the tetrad-specific
shocks $\varepsilon_{E_t}$ do not affect any link outcome outside the
tetrad.  Under the strategic-neighborhood framework of
Section~\ref{sec:primitive}, a sufficient condition is that no agent
outside the tetrad belongs to a non-robustness component containing
a tetrad agent: formally,
$C_a\cap\{i,j,h,k\}^c=\emptyset$ for each $a\in\{i,j,h,k\}$,
where~$C_a$ is agent~$a$'s connected component in the non-robustness
graph~$\bm{D}$.  This ensures that all non-robust links involving
tetrad agents are links \emph{among} the four tetrad agents
themselves, so perturbing $\varepsilon_{E_t}$ cannot propagate outside
the tetrad.  Note that this condition does \emph{not} require the four
agents' strategic neighborhoods to be mutually disjoint---agents $i$
and~$j$ may share robust neighbors, and the link $Y_{ij}$ itself may
be non-robust, provided the non-robustness component containing~$i$
and~$j$ lies entirely within $\{i,j,h,k\}$.  Under the
bounded-strategic-neighborhood condition
(Assumption~\ref{ass:bdd_nbhd}), this isolation condition holds for a
positive fraction of quadruplets.

\begin{assumption}[Local Tetrad Uniqueness]\label{ass:local_uniq}
  For any admissible tetrad $(i,j,h,k)\in\mathcal{T}_n$, conditional
  on $(Z,A,Y_{-E_t})$, if the tetrad pattern~$T_t$ (resp.\ the flipped
  pattern~$F_t$) satisfies the best-response conditions for all four
  links $e\in E_t$, then no other tetrad configuration is
  simultaneously an equilibrium.
\end{assumption}

Assumption~\ref{ass:local_uniq} rules out equilibrium multiplicity
\emph{within the tetrad} on the comparison sample $T_t\cup F_t$.
Note that $T_t$ and~$F_t$ themselves can never coexist as equilibria
under Assumption~\ref{ass:cpi}: their best-response conditions
require the shock~$\varepsilon_e$ to fall on opposite sides of the
same threshold~$\tilde\delta_e$ for every $e\in E_t$, so the two
events are mutually exclusive.  What Assumption~\ref{ass:local_uniq}
additionally rules out is coexistence of~$T_t$ (or~$F_t$) with a
\emph{third} tetrad pattern.

\begin{remark}[Sufficient Conditions for Local Tetrad Uniqueness]%
  \label{rem:local_uniq}
  \leavevmode
  \begin{enumerate}[label=(\roman*)]
  \item \emph{Common friends
    ($X_{ij}=\sum_l Y_{il}Y_{jl}$).}
    Under Assumption~\ref{ass:cpi},
    every within-quadruplet contribution to~$X_e$ vanishes on
    $T_t\cup F_t$; hence the threshold $\tilde\delta_e$ does not
    depend on other tetrad links at all.  The four links then have no
    within-tetrad strategic interaction, and each link is determined
    independently by its own shock.  Local uniqueness is automatic.
  \item \emph{Standard Jaccard with $\gamma_0\ge 0$.}
    When
    $\gamma_0\ge 0$, Jaccard's denominator is weakly decreasing in the
    number of within-tetrad links, so a larger denominator lowers the
    covariate value and hence the threshold.  A monotonicity argument
    shows that any best-response--consistent pattern is uniquely
    pinned down: if a link is present in~$T_t$ but absent in some
    competing pattern~$P$, the threshold under~$P$ is weakly higher,
    so the shock that justified forming the link under~$T_t$ still
    justifies it under~$P$---contradicting the absence.  The symmetric
    argument applies when a link is absent in~$T_t$ but present
    in~$P$.
  \end{enumerate}
\end{remark}

\begin{definition}[Admissible Tetrad]\label{def:admissible}
  A tetrad $(i,j,h,k)$ is \emph{admissible} if the two diagonal links
  are absent:
  \[
    Y_{ih}=0\quad\text{and}\quad Y_{jk}=0.
  \]
  The set of all admissible tetrads satisfying
  Assumptions~\ref{ass:tetrad_exog} and~\ref{ass:local_uniq} is denoted $\mathcal{T}_n$.
\end{definition}

In a sparse network with expected degree $O(1)$, each link exists with
probability $O(n^{-1})$, so the diagonal-absence condition holds with
probability approaching~$1$ for any given quadruplet.  Consequently,
the number of admissible tetrads is $\Theta(n^4)$, and after the
greedy packing argument of Proposition~\ref{prop:packing} one retains
$\Theta(n)$ independent tetrads---the same asymptotic order as in the
partial-identification setting.

\subsection{Point Identification Result}

The point identification result requires a
strengthened version of Assumption~\ref{ass:tetrad_id} that
conditions on the realized endogenous covariates and the diagonal link
absences.

\begingroup
\renewcommand{\theassumption}{4$''$}
\begin{assumption}[Augmented Tetrad Estimation]\label{ass:tetrad_id_aug}
  The conclusion of Assumption~\ref{ass:tetrad_id} continues to hold
  when the conditioning is augmented to $(Z_S,\, X_t,\, Y_{ih}=0,\, Y_{jk}=0)$, where $X_t$ is the tetrad endogenous covariate vector defined in~\eqref{eq:Xt}.
\end{assumption}
\endgroup

The tetrad differencing
eliminates all individual fixed effects \emph{algebraically}---exactly as in the
$\gamma_0=0$ case---while the endogenous covariates are
handled by the isolation conditioning.  The two mechanisms do not
interfere with each other: fixed effects cancel regardless of the endogenous
covariates, and endogenous covariates decouple regardless of the
fixed effects.

\begin{theorem}[Point Identification with Endogenous Covariates and Fixed Effects]%
  \label{thm:point_id}
  Suppose Assumptions~\ref{ass:random_sampling}--\ref{ass:iid_errors},
  \ref{ass:cpi}--\ref{ass:local_uniq}, and~\ref{ass:tetrad_id_aug} hold, and
  $\varepsilon_{ij}\overset{\mathrm{iid}}{\sim}\mathrm{Logistic}(0,1)$
  independently of $(Z,A)$.
  For any admissible tetrad $(i,j,h,k)\in\mathcal{T}_n$, define
  \begin{align}
    p_{+}(\zeta,x_t) &:= \Prob\!\bigl(Y_{ij}Y_{hk}(1-Y_{ik})(1-Y_{jh})=1
      \bigm| \zeta,\, X_t=x_t,\, Y_{ih}=0,\, Y_{jk}=0\bigr),\label{eq:pplus}\\
    p_{-}(\zeta,x_t) &:= \Prob\!\bigl((1-Y_{ij})(1-Y_{hk})Y_{ik}Y_{jh}=1
      \bigm| \zeta,\, X_t=x_t,\, Y_{ih}=0,\, Y_{jk}=0\bigr).\label{eq:pminus}
  \end{align}
  Then for every admissible tetrad $(i,j,h,k)\in\mathcal{T}_n$ and
  every $(\zeta,x_t)$ in the support,
  \begin{equation}\label{eq:logodds_fe}
    \log\frac{p_+(\zeta,x_t)}{p_-(\zeta,x_t)}
    \;=\;
    \Delta Z(\zeta)'\beta_0+\Delta X(x_t)'\gamma_0,
  \end{equation}
  where $\Delta Z:=Z_{ij}+Z_{hk}-Z_{ik}-Z_{jh}$ and
  $\Delta X:=X_{ij}+X_{hk}-X_{ik}-X_{jh}$.
  If the support of $(\Delta Z,\Delta X)$ restricted to admissible
  tetrads spans $\R^{d_\beta+d_\gamma}$, then
  $\theta_0=(\beta_0,\gamma_0)$ is point identified.
\end{theorem}

The proof is given in Appendix~\ref{app:proofs}.

\subsection{Constructive Estimator}

The proof of Theorem~\ref{thm:point_id} is constructive and yields a
computationally simple estimator.
\begin{enumerate}[label=(\roman*)]
  \item \emph{Select admissible tetrads.} From the observed network,
    enumerate quadruplets $(i,j,h,k)$ satisfying $Y_{ih}=0$ and
    $Y_{jk}=0$.
  \item \emph{Classify.} For each admissible tetrad, record whether the
    tetrad pattern ($Y_{ij}Y_{hk}(1-Y_{ik})(1-Y_{jh})=1$) or the
    flipped pattern ($(1-Y_{ij})(1-Y_{hk})Y_{ik}Y_{jh}=1$) obtains;
    discard tetrads with neither pattern.
  \item \emph{Conditional logit.} Among the retained tetrads, run a
    conditional logit regression of the binary outcome
    ($\text{tetrad}=1$, $\text{flipped}=0$) on the regressors
    $(\Delta Z,\Delta X)$:
    \begin{equation}\label{eq:condlogit}
      \Prob(\text{tetrad}\mid\text{tetrad or flipped},\;\zeta,\;x_t)
      =\frac{\exp\!\bigl(\Delta Z'\beta+\Delta X'\gamma\bigr)}
      {1+\exp\!\bigl(\Delta Z'\beta+\Delta X'\gamma\bigr)}.
    \end{equation}
    The maximum likelihood estimator
    $\hat\theta=(\hat\beta,\hat\gamma)$ consistently estimates
    $\theta_0$.
\end{enumerate}

\begin{remark}[Comparison with the $\gamma_0=0$ Case and \citet{graham2017econometric}]\label{rem:comparison}
  Theorem~\ref{thm:point_id} generalizes the tetrad logit of
  \citet{graham2017econometric} in two directions: it allows endogenous
  covariates ($\gamma_0\neq 0$), and it uses unconditional (marginal
  over~$A$) tetrad probabilities rather than probabilities conditional
  on the sufficient statistic for~$A$.  When $\gamma_0=0$, there are no
  endogenous covariates, the diagonal-absence condition is vacuous
  (since $X_e\equiv 0$ does not depend on any link), and
  Theorem~\ref{thm:point_id} reduces to the standard tetrad logit with
  $\log[p_+(\zeta)/p_-(\zeta)]=\Delta Z'\beta_0$.  The additional
  structural restrictions (Assumptions~\ref{ass:cpi}--\ref{ass:local_uniq} and diagonal absence)
  are the price paid for accommodating $\gamma_0\neq 0$.
\end{remark}

\begin{remark}[Abundance of Admissible Tetrads]\label{rem:abundance}
  In a sparse network with expected degree $O(1)$, $\Prob(Y_{ih}=0)=1-O(n^{-1})\to 1$
  and similarly for $Y_{jk}$.  Hence the diagonal-absence condition
  is satisfied with probability approaching $1$ for any given quadruplet, and the number of admissible tetrads is $\Theta(n^4)$.  The
  additional requirement that no outside agent shares a non-robustness
  component with a tetrad agent (for
  Assumption~\ref{ass:tetrad_exog}) is closely related to the condition used in
  the packing argument of Proposition~\ref{prop:packing}; under
  Assumption~\ref{ass:bdd_nbhd}, one retains $\Theta(n)$ independent
  admissible tetrads.
\end{remark}

\begin{remark}[Scope of the Comparison-Pattern Invariance Condition]%
  \label{rem:cpi_scope}
  Assumption~\ref{ass:cpi} covers the two most empirically important
  endogenous covariates: common friends
  $\mathrm{CF}_{ij}=\sum_{l\neq i,j}Y_{il}Y_{jl}$ and the standard
  Jaccard index
  $\mathrm{CF}_{ij}/|N(i)\cup N(j)|$
  (Examples~\ref{ex:cf_cpi}--\ref{ex:jac_cpi}).
  The bilinear condition
  ($X_{ij}=\psi(\{Y_{il}Y_{jl}\}_{l\neq i,j},Z)$) is a special case
  that implies Assumption~\ref{ass:cpi} globally, not just on
  $T_t\cup F_t$.
  More generally, Assumption~\ref{ass:cpi} is satisfied by any
  $\phi_{ij}(Y_{-ij},Z)$ whose within-tetrad dependence runs
  exclusively through the diagonal links (which are zero on admissible
  tetrads) and treats the two comparison patterns symmetrically
  in their effect on each endpoint's local neighborhood statistics.
  For the partial identification results, the bounding-by-$c$
  technique of
  Theorems~\ref{thm:tetrad}--\ref{thm:tetrad_parametric} remains
  valid for arbitrary~$\phi_{ij}$ without any comparison-pattern
  restriction.
\end{remark}

\begin{remark}[Rank Condition]\label{rem:rank}
  The rank condition requires that $(\Delta Z,\Delta X)$ has
  $(d_\beta+d_\gamma)$-dimensional variation across admissible tetrads.
  Since $\Delta X=\mathrm{CF}_{ij}+\mathrm{CF}_{hk}-\mathrm{CF}_{ik}-\mathrm{CF}_{jh}$
  (under isolation) depends on the local network structure around the
  tetrad agents, it provides genuine variation beyond $\Delta Z$ as
  long as the agents have heterogeneous local neighborhoods.  When all
  common-friend counts are zero (e.g., in an extremely sparse
  network), $\Delta X\equiv 0$ and only~$\beta_0$ can be identified;
  the estimator naturally reduces to the standard tetrad logit.
  In networks with nontrivial clustering, $\Delta X$ takes many
  distinct values, ensuring the rank condition is generically
  satisfied.
\end{remark}

\section{Simulation Based on Tetrad Restrictions}

\label{sec:simulation} 

We examine the finite-sample behavior of the tetrad-based partial identification
approach through simulation, focusing on how the network size
$n$, unobserved heterogeneity (fixed effects) $A$,
the support size of the discrete exogenous covariate $Z$,
and the endogenous covariate $X$ affect the sharpness
of the identified set.
For simplicity, we use only the baseline tetrad restrictions (Theorems~\ref{thm:tetrad}--\ref{thm:tetrad_parametric}), leaving the longer-cycle restrictions and aggregation over $\mathcal{W}$ for future work.


\subsection{Model Specification and Data Generating Process}

\label{subsec:dgp} 

We conduct simulation studies under three nested specifications, which
differ by whether individual fixed effects and/or endogenous covariates
are present.

\paragraph{Baseline Model.}

The individual-level exogenous characteristics $Z_{i}$ is a two-dimensional
vector $(Z_{1i},Z_{2i})\in[-10,10]^{2}$, and the vector of exogenous
dyadic covariates is constructed as $Z_{ij}=|Z_{i}-Z_{j}|$. The true
link formation equation is
\begin{equation}
  Y_{ij}=\ind\left\{ Z_{ij,1}+\gamma_{0}Z_{ij,2}+\varepsilon_{ij}\geq0\right\} .\label{eq:link_formation_baseline}
\end{equation}
Under all three specifications, the shocks $\varepsilon_{ij}$ are
i.i.d. following a $\text{Logistic}(0,1)$ distribution.

\paragraph{FE-only Model.}

We add individual-level unobserved fixed effects $A_{i}\sim\mathcal{N}(\rho Z_{1i},\sigma_{A}^{2})$
to the baseline model, where $\rho\in[0,1]$ captures the correlation
between unobserved heterogeneity and observed characteristics. The
true link formation equation becomes
\begin{equation}
  Y_{ij}=\ind\left\{ Z_{ij,1}+\gamma_{0}Z_{ij,2}+A_{i}+A_{j}+\varepsilon_{ij}\geq0\right\} .\label{eq:link_formation_fe}
\end{equation}
Note that this is equivalent to the main model~\eqref{eq:link_formation} with $\tilde A_i:=-A_i$ playing the role of the fixed effect (and $-\varepsilon_{ij}$ for the shock, which has the same logistic distribution by symmetry).\footnote{Under this reparametrization, $\tilde A_i\sim\mathcal{N}(-\rho Z_{1i},\sigma_A^2)$ and equation~\eqref{eq:link_formation_fe} becomes $Y_{ij}=\ind\{Z_{ij,1}+\gamma_0 Z_{ij,2}+\varepsilon_{ij}\ge\tilde A_i+\tilde A_j\}$, matching the sign convention in~\eqref{eq:link_formation}.  The tetrad restrictions are invariant to this reparametrization since $A_i$ cancels in all cases.}

\paragraph{Full Model.}

We further add an endogenous dyadic covariate $X_{ij}$ to the model
while reducing the dimension of $Z_{ij}$ to 1, i.e., $Z_{i}\in\mathcal{Z}=[-10,10]$
and $Z_{ij}=|Z_{i}-Z_{j}|$. We set the coefficient for $Z_{ij}$
to be $\beta_{0}=1$, fixed and known. The endogenous covariate is the Jaccard index of common friends:
\[
  X_{ij} = \frac{|N(i)\cap N(j)|}{|N(i)\cup N(j)|},
\]
where $N(i)=\{k:\,Y_{ik}=1\}$ denotes the neighbor set of $i$.
The true link formation equation becomes
\begin{equation}
  Y_{ij}=\ind\left\{ Z_{ij}+\gamma_{0}X_{ij}+A_{i}+A_{j}+\varepsilon_{ij}\geq0\right\} .\label{eq:link_formation_full}
\end{equation}
Under this model, the probabilities involved in the identifying restriction
\eqref{eq:tetrad_main} no longer have a closed form. Therefore, we
need to approximate the probabilities using a realized network.

In order to generate a realized network, we first discretize $\mathcal{Z}$
and generate $Z_{i}$ from $\mathcal{Z}$ by randomly sampling with
equal probabilities. Since $X_{ij}$ depends on the realized network
$Y$, it requires solving for a network consistent with \eqref{eq:link_formation_full}
given the shocks and fixed effects. We set $A_{i}\overset{\text{i.i.d}}{\sim}\mathcal{N}(0,1)$
and generate data by iterating a link-update procedure (starting from
a zero $Y$ matrix and recomputing $X_{ij}$ after each update) until
convergence. The parameter to be identified is $\gamma_{0}$.

\begin{remark}[Simulation Design Choices]\label{rem:sim_design}
  Several design choices merit discussion. The value $\gamma_{0}=4$
  was chosen to ensure that the endogenous covariate has a substantial
  effect on link formation, making the strategic interaction component
  quantitatively important. Since $X_{ij}\in[0,1]$,
  a coefficient of $\gamma_{0}=4$ means that dyads with identical neighborhoods
  ($X_{ij}=1$) receive a boost of $4$ in the latent
  index, comparable to the standard deviation of the logistic shock
  ($\pi/\sqrt{3}\approx1.81$). Results for moderate values of $\gamma_{0}$
  may yield tighter identified sets.

  The best-response iteration starting from the empty network selects
  one particular equilibrium when multiple equilibria exist. We do not
  verify uniqueness of the equilibrium; thus our simulation results
  are conditional on this particular equilibrium selection rule. The
  identified set derived from our restrictions is valid regardless of
  equilibrium selection (Theorem~\ref{thm:tetrad}), but the finite-sample
  informativeness of the criterion may vary across equilibria.
\end{remark}


We use $\gamma$ to denote the parameter to be identified. Following
Theorem~\ref{thm:tetrad}, we evaluate the (sample analogue of the)
tetrad criterion $Q^{\mathrm{tetrad}}(\gamma)$ defined in
\eqref{eq:tetrad_criterion}, where the relevant conditional probabilities
$p_{L}$ and $p_{U}$ are computed (i) in closed form under the baseline model,
(ii) by Monte Carlo integration over fixed effects under the FE-only model,
and (iii) by sample proportions under the full model.
The identified set is the set of values for which the criterion is non-positive, i.e. $Q^{\mathrm{tetrad}}(\gamma) \leq 0.$

Under the full model, we also implement the stronger criterion that
enforces both sides of the tetrad restriction under the parametric assumption on the logistic distribution of $\varepsilon_{ij}$, as in Theorem~\ref{thm:tetrad_parametric}; the resulting
identified set is denoted by $\Gamma_{\mathrm{strict}}$.

We evaluate the criterion functions on a finite grid of $\gamma$ values. The supremum
is computed by the GenSA algorithm under the baseline model and the
FE-only model, while it is computed by grid-search under the full
model.


\subsection{Simulation Results}

\label{subsec:sim_results} 


\subsubsection{Baseline Model}

\label{subsec:results_baseline} 

Figure \ref{fig:baseline} shows how the value of the main criterion
function $Q(\gamma)$ changes with $\gamma$ between $-10$ and $10$
when the true parameter is $\gamma_{0}=1$. In the figure, $Q$ is plotted after shifting by $+1$, so the identified set $\{\gamma:Q(\gamma)\le0\}$ corresponds to the region where the plotted curve attains its minimum value of~$1$. The identified set for
$\gamma$ is $[1,5]$.

\begin{figure}[htbp]
  \centering{}\includegraphics[scale=0.5]{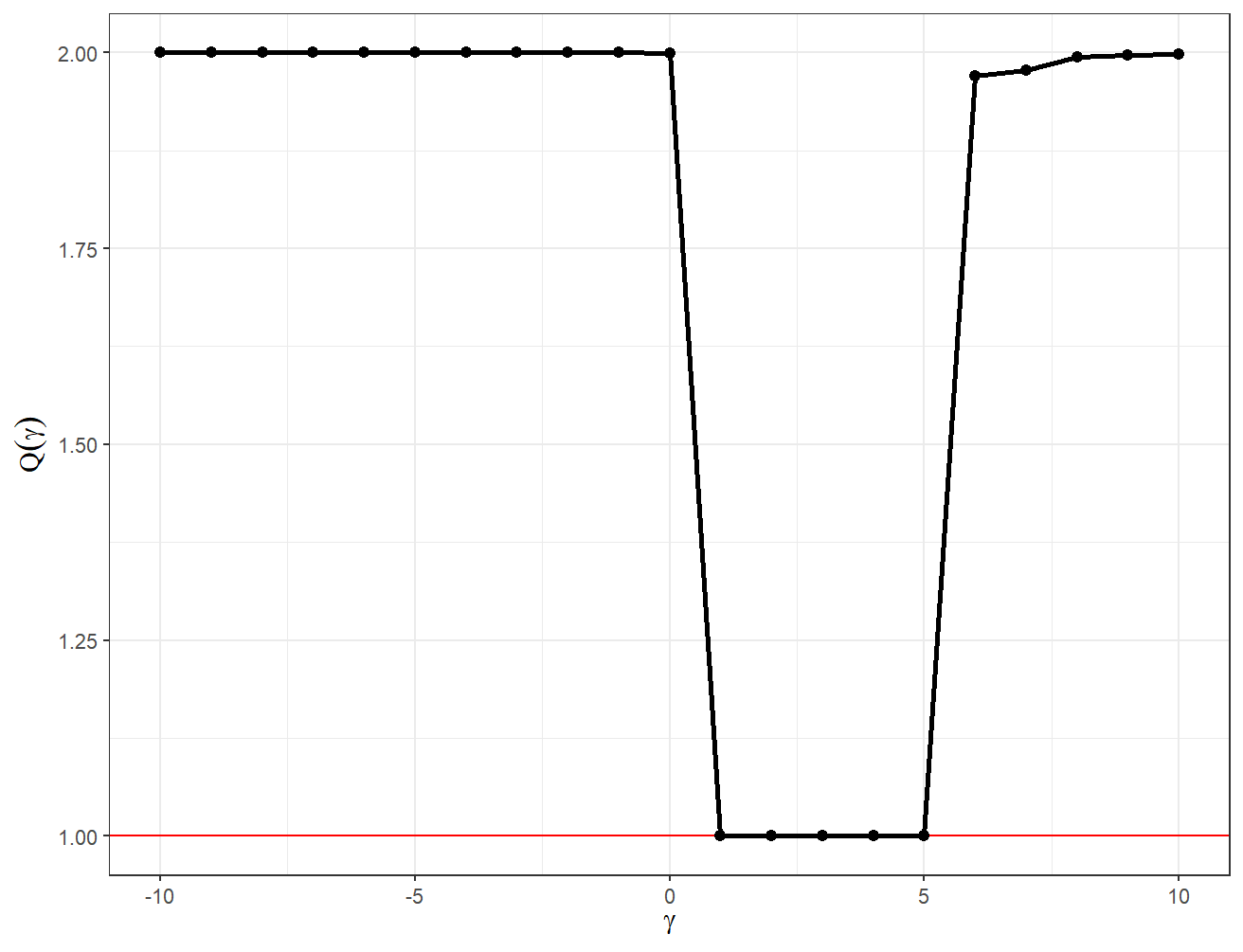}
  \caption{Baseline criterion $Q$ as a function of $\gamma$.}
  \label{fig:baseline}
\end{figure}


\subsubsection{FE-only Model}

\label{subsec:results_fe_only} 

Table \ref{tab:FE_only} reports the identifying results for $\gamma$
under the FE-only model \eqref{eq:link_formation_fe} with fixed-effect
designs that vary (i) the dispersion of individual fixed effects and
(ii) the strength of correlation between individual fixed effects
and observed individual characteristics $Z_{i}$.

\begin{table}[htbp]
  \centering \caption{FE-only model: Identified sets for $\gamma$ under alternative fixed-effect
  designs.}
  \label{tab:FE_only}
  \global\long\def\arraystretch{1.2}%
  \setlength{\tabcolsep}{8pt} %
  \begin{tabular}{lcccc}
    \hline
    & (1)  & (2)  & (3)  & (4) \tabularnewline
    \hline
    Correlation with $Z_{1i}$  & No  & No  & Weak  & Strong \tabularnewline
    SD of fixed effects ($\sigma_{A}$)  & $1$  & $5$  & $1$  & $1$ \tabularnewline
    Identified set for $\gamma$  & $[1,7]$  & $[0,\infty)$  & $[1,6]$  & $[1,\infty)$ \tabularnewline
    \hline
  \end{tabular}

  \vspace{3pt}
  \textit{\footnotesize{}Notes.}{\footnotesize{} The identified set
    is $\{\gamma:Q(\gamma)\le0\}$. In correlated designs, $A_{i}=\rho Z_{1i}+u_{i}$
    with $u_{i}\sim\mathcal{N}(0,\sigma_{A}^{2})$; ``no'', ``weak'',
    and ``strong'' correlation means $\rho=0$, $\rho=0.1$, and $\rho=0.9$,
  respectively.}
\end{table}

Under the FE-only model, introducing unobserved heterogeneity substantially
weakens identification of $\gamma$ relative to the baseline model
without fixed effects. In particular, while the baseline model yields
a finite identified set $[1,5]$ when the true value is $\gamma_{0}=1$,
the FE-only model produces noticeably wider sets that become sensitive
to both the dispersion and endogeneity of fixed effects. When fixed
effects are uncorrelated with individual characteristics and have
small dispersion, the identified set remains bounded above but expands
a little compared to the baseline; increasing heterogeneity
makes the upper bound disappear, yielding one-sided identification.
Holding dispersion fixed, correlation between the fixed effects and
individual characteristics also weakens identification when the correlation
is strong. Overall, unobserved heterogeneity, especially when strongly
correlated with observables or with large dispersion, reduces the
informativeness of the identifying criterion and tends to eliminate
finite upper bounds on $\gamma$, but the sign of $\gamma$ can still
be identified in all cases.


\subsubsection{Full Model}

\label{subsec:results_full} 

Table~\ref{tab:full_jaccard} reports the identified sets obtained
from both the main criterion $Q(\gamma)\le0$ and the strict criterion
(Theorem~\ref{thm:tetrad_parametric}), using the Jaccard index as the
endogenous covariate with $n=100$ and a single network draw.

\begin{table}[htbp]
  \centering
  \caption{Full model: Identified sets for $\gamma$ ($n=100$, $\gamma_{0}=4$, Jaccard index).}
  \label{tab:full_jaccard}
  \global\long\def\arraystretch{1.15}%
  \setlength{\tabcolsep}{8pt}%
  \begin{tabular}{c cc}
    \hline
    $|\mathcal{Z}|$ & $Q(\gamma)\le0$ & Strict \tabularnewline
    \hline
    $3$  & $(-\infty,\,28]$ & $(-\infty,\,24]$   \tabularnewline
    $9$  & $(-\infty,\,20]$ & $[-26,\,18]$        \tabularnewline
    $15$ & $(-\infty,\,20]$ & $[-5,\,19]$         \tabularnewline
    $21$ & $(-\infty,\,17]$ & $\mathbf{[4,\,11]}$ \tabularnewline
    \hline
  \end{tabular}

  \vspace{3pt}
  \noindent
  \begin{minipage}[c]{1\linewidth}%
    \raggedright \textit{\footnotesize{}Notes.}{\footnotesize{}
      $(-\infty,\,\bar\gamma]$ indicates a lower endpoint at the
      search-grid boundary $-40$. Results are based on a single network draw.} %
  \end{minipage}
\end{table}

The identified set tightens as the support size $|\mathcal{Z}|$
grows, since more exogenous variation enables more informative tetrad comparisons. At $|\mathcal{Z}|=21$, the
strict identified set is $[4,11]$, a bounded interval of width $7$
containing the true value $\gamma_{0}=4$. This shows that
the tetrad restrictions can produce nontrivial bounds on
the strategic-interaction parameter even with
endogenous covariates and unobserved individual fixed effects.

These results are preliminary: they are based on a single network draw at
$n=100$, and the current design fixes $\beta_0$ at its true value and searches only over $\gamma$, which does not address the challenges of joint identification over the full parameter vector $(\beta,\gamma)$. A more systematic Monte Carlo study---varying network sizes, averaging across repeated draws, searching over the full parameter space, and
constructing formal confidence sets for the identified set (e.g.,
along the lines of \citealt{andrews2013inference} and
\citealt*{chernozhukov2013intersection})---is under
investigation and will be reported in a subsequent version of this
paper.

\section{Conclusion}\label{sec:discussion}

This paper develops a tractable identification approach for strategic network formation models with endogenous network statistics and unobserved individual fixed effects. The main idea is a ``bounding-by-$c$'' construction applied to subnetwork configurations (tetrads, triads, and more general weighted cycles), which produces moment-inequality restrictions without requiring characterization of the equilibrium mapping. Section~\ref{sec:primitive} gives primitive conditions under which the high-level Assumption~\ref{ass:tetrad_id} holds, embedding the model in the sparse network framework of \citet{leung2019treatment} and establishing consistency of the tetrad conditional probability estimator via a greedy packing argument.

Several directions remain for future work. First, the central limit theorems of \citet{leung2019normal} and \citet{menzel2021clt} can be used to develop formal inference procedures for the identified sets, building on the econometric theory of inference based on conditional moment inequalities \citep*{andrews2013inference,chernozhukov2013intersection}. Second, a more systematic set of Monte Carlo simulations would provide a fuller picture of the finite-sample performance of both the set-identification approach and the point identification result, together with the associated estimator. Third, an empirical application using real-world network data is a natural next step.

\bibliographystyle{ecta}
\bibliography{GLX_StratNet}

\begin{thebibliography}{22}
\newcommand{\enquote}[1]{``#1''}
\expandafter\ifx\csname natexlab\endcsname\relax\def\natexlab#1{#1}\fi

\bibitem[\protect\citeauthoryear{Aguirregabiria and Mira}{Aguirregabiria and Mira}{2007}]{aguirregabiria2007sequential}
\textsc{Aguirregabiria, V. and P.~Mira} (2007): \enquote{Sequential estimation of dynamic discrete games,} \emph{Econometrica}, 75, 1--53.

\bibitem[\protect\citeauthoryear{Andrews and Shi}{Andrews and Shi}{2013}]{andrews2013inference}
\textsc{Andrews, D.~W. and X.~Shi} (2013): \enquote{Inference based on conditional moment inequalities,} \emph{Econometrica}, 81, 609--666.

\bibitem[\protect\citeauthoryear{Bajari, Benkard, and Levin}{Bajari et~al.}{2007}]{bajari2007estimating}
\textsc{Bajari, P., C.~L. Benkard, and J.~Levin} (2007): \enquote{Estimating dynamic models of imperfect competition,} \emph{Econometrica}, 75, 1331--1370.

\bibitem[\protect\citeauthoryear{Bonhomme}{Bonhomme}{2012}]{bonhomme2012functional}
\textsc{Bonhomme, S.} (2012): \enquote{Functional Differencing,} \emph{Econometrica}, 80, 1337--1385.

\bibitem[\protect\citeauthoryear{Bonhomme and Dano}{Bonhomme and Dano}{2023}]{bonhomme2023functional}
\textsc{Bonhomme, S. and K.~Dano} (2023): \enquote{Functional Differencing in Networks,} \emph{Revue {\'E}conomique}.

\bibitem[\protect\citeauthoryear{Bonhomme, Dano, and Graham}{Bonhomme et~al.}{2025}]{bonhomme2025moment}
\textsc{Bonhomme, S., K.~Dano, and B.~S. Graham} (2025): \enquote{Moment Restrictions for Nonlinear Panel Data Models with Feedback,} \emph{arXiv preprint arXiv:2506.12569}.

\bibitem[\protect\citeauthoryear{Chernozhukov, Lee, and Rosen}{Chernozhukov et~al.}{2013}]{chernozhukov2013intersection}
\textsc{Chernozhukov, V., S.~Lee, and A.~M. Rosen} (2013): \enquote{Intersection bounds: Estimation and inference,} \emph{Econometrica}, 81, 667--737.

\bibitem[\protect\citeauthoryear{Dano, Honor{\'e}, and Weidner}{Dano et~al.}{2025}]{dano2025binary}
\textsc{Dano, K., B.~E. Honor{\'e}, and M.~Weidner} (2025): \enquote{Binary Choice Logit Models with General Fixed Effects for Panel and Network Data,} \emph{arXiv preprint arXiv:2508.11556}.

\bibitem[\protect\citeauthoryear{de~Paula}{de~Paula}{2020}]{de2020econometric}
\textsc{de~Paula, {\'A}.} (2020): \enquote{Econometric models of network formation,} \emph{Annual Review of Economics}, 12, 775--799.

\bibitem[\protect\citeauthoryear{de~Paula, Richards-Shubik, and Tamer}{de~Paula et~al.}{2018}]{de2018identifying}
\textsc{de~Paula, {\'A}., S.~Richards-Shubik, and E.~Tamer} (2018): \enquote{Identifying preferences in networks with bounded degree,} \emph{Econometrica}, 86, 263--288.

\bibitem[\protect\citeauthoryear{Dzemski}{Dzemski}{2019}]{dzemski2019empirical}
\textsc{Dzemski, A.} (2019): \enquote{An empirical model of dyadic link formation in a network with unobserved heterogeneity,} \emph{Review of Economics and Statistics}, 101, 763--776.

\bibitem[\protect\citeauthoryear{Gao}{Gao}{2020}]{gao2020}
\textsc{Gao, W.~Y.} (2020): \enquote{Nonparametric Identification in Index Models of Link Formation,} \emph{Journal of Econometrics}, 215, 399--413.

\bibitem[\protect\citeauthoryear{Gao, Li, and Xu}{Gao et~al.}{2023}]{gao2023logical}
\textsc{Gao, W.~Y., M.~Li, and S.~Xu} (2023): \enquote{Logical differencing in dyadic network formation models with nontransferable utilities,} \emph{Journal of Econometrics}, 235, 302--324.

\bibitem[\protect\citeauthoryear{Gao and Wang}{Gao and Wang}{2026}]{gao2026identification}
\textsc{Gao, W.~Y. and R.~Wang} (2026): \enquote{Identification in nonlinear dynamic panel models under partial stationarity,} \emph{Journal of Econometrics}, 253, 106185.

\bibitem[\protect\citeauthoryear{Graham}{Graham}{2017}]{graham2017econometric}
\textsc{Graham, B.~S.} (2017): \enquote{An Econometric Model of Network Formation with Degree Heterogeneity,} \emph{Econometrica}, 85, 1033--1063.

\bibitem[\protect\citeauthoryear{Jackson and Watts}{Jackson and Watts}{2002}]{jackson2002evolution}
\textsc{Jackson, M.~O. and A.~Watts} (2002): \enquote{The evolution of social and economic networks,} \emph{Journal of economic theory}, 106, 265--295.

\bibitem[\protect\citeauthoryear{Leung}{Leung}{2019}]{leung2019treatment}
\textsc{Leung, M.~P.} (2019): \enquote{A weak law for moments of pairwise-stable networks,} \emph{Journal of Econometrics}, 210, 310--326.

\bibitem[\protect\citeauthoryear{Leung and Moon}{Leung and Moon}{2025}]{leung2019normal}
\textsc{Leung, M.~P. and H.~R. Moon} (2025): \enquote{Normal approximation in large network models,} \emph{arXiv preprint arXiv:1904.11060v5}.

\bibitem[\protect\citeauthoryear{Mele}{Mele}{2017}]{mele2017structural}
\textsc{Mele, A.} (2017): \enquote{A Structural Model of Homophily and Clustering in Social Networks,} Working paper.

\bibitem[\protect\citeauthoryear{Menzel}{Menzel}{2021}]{menzel2021clt}
\textsc{Menzel, K.} (2021): \enquote{Central Limit Theory for Models of Strategic Network Formation,} Working paper, New York University.

\bibitem[\protect\citeauthoryear{Menzel}{Menzel}{2026}]{menzel2026strategic}
---\hspace{-.1pt}---\hspace{-.1pt}--- (2026): \enquote{Strategic network formation with many agents,} \emph{Journal of Econometrics}, 253, 106174.

\bibitem[\protect\citeauthoryear{Sheng}{Sheng}{2020}]{sheng2020structural}
\textsc{Sheng, S.} (2020): \enquote{A structural econometric analysis of network formation games through subnetworks,} \emph{Econometrica}, 88, 1829--1858.

\end{thebibliography}

\appendix
\section{Proofs}\label{app:proofs}

\subsection{Proof of Proposition~\ref{prop:packing}}
\begin{proof} We prove the result using a packing argument. 
   For a tetrad $t=(i,j,h,k)$, define its \emph{dependence set}
  \[
    D_t \;:=\; \bigcup_{a\in\{i,j,h,k\}} C_a^+.
  \]
  This is the set of all agents whose primitives can affect the tetrad
  outcome~$E_t$.  Indeed, by
  Assumption~\ref{ass:decentral-lm}, each link among
  $\{i,j,h,k\}$ is determined by the equilibrium on the strategic
  neighborhood containing its endpoints, and all such neighborhoods
  are contained in~$D_t$.
  Since the connected components $C_a$ of the non-robustness graph
  partition $\mathcal{N}_n$, each agent belongs to exactly one
  component; however, the strategic neighborhoods $C_a^+$ may
  overlap.  By the union bound and
  Assumption~\ref{ass:bdd_nbhd},
  $\mathbb{E}[|D_t|]\le \sum_{a\in t}\mathbb{E}[|C_a^+|]=O(1)$.

  Define a \emph{conflict graph} $G_D$ on the set of all tetrads with
  $\zeta_t=\zeta$: two tetrads $t,t'$ are adjacent in $G_D$ if and
  only if $D_t\cap D_{t'}\ne\emptyset$.  Since
  $D_{t'}=\bigcup_{a\in t'}C_a^+$, the tetrads $t$ and $t'$ conflict
  if and only if at least one agent $a\in t'$ satisfies
  $C_a^+\cap D_t\ne\emptyset$.  Define the \emph{second-order
  dependence set}
  \[
    D_t^{(2)}\;:=\;\bigl\{a\in\mathcal{N}_n:C_a^+\cap D_t\ne\emptyset\bigr\}.
  \]

  We now bound the two key quantities:
  \begin{itemize}
    \item \emph{Number of vertices.}
      Since augmented types $\tilde Z_i=(Z_i,A_i)$ are i.i.d.\
      (Assumption~\ref{ass:random_sampling}) and $z_S$ lies in the finite support
      of~$Z_S$, the number of tetrads with $Z_S=z_S$ is
      $N=\binom{n}{4}\mathbb{P}(Z_S=z_S)=\Theta(n^4)$.
    \item \emph{Average degree.}
      A tetrad $t'$ conflicts with $t$ only if at least one agent of
      $t'$ lies in~$D_t^{(2)}$.  Under the subcriticality condition of
      \citet{leung2019treatment}, $|C_i^+|$ has exponentially
      decaying tails; iterating the branching-process bound once
      more shows that $|D_t^{(2)}|$ also has bounded
      expectation: $\mathbb{E}[|D_t^{(2)}|]=O(1)$.
      Each agent in $D_t^{(2)}$ can appear in at most $\binom{n-1}{3}$
      other tetrads, so the expected degree of~$t$ in $G_D$ satisfies
      $\mathbb{E}[\deg(t)]\le\mathbb{E}[|D_t^{(2)}|]\cdot n^3=O(n^3)$.
      Since this bound holds for every tetrad~$t$, the average degree
      in $G_D$ is $\bar d=O(n^3)$.
  \end{itemize}
  By the Tur\'an bound, $G_D$ contains an
  independent set of size at least $N/(\bar d+1)\ge
  \Theta(n^4)/O(n^3)=\Theta(n)$.
  That is, the greedy algorithm---which sequentially selects any
  tetrad not conflicting with those previously selected---packs at
  least $m_n=\Theta(n)$ tetrads $t_1,\ldots,t_{m_n}$ with pairwise
  disjoint dependence sets.  (The greedy selection is deterministic
  given the realization of $(Z,A,\varepsilon)$, so the selected
  tetrads are random.  The key probabilistic
  properties---independence and identical conditional means---hold
  for these selected tetrads as shown below.)

  Since $D_{t_s}\cap D_{t_{s'}}=\emptyset$ for $s\ne s'$, the
  primitives determining each $E_{t_s}$ are disjoint subcollections
  of i.i.d.\ random variables
  ($\tilde Z_i$ across agents, $\varepsilon_{ij}$ across pairs).
  By Assumption~\ref{ass:decentral-lm}, the equilibrium on each
  strategic neighborhood is determined only by the primitives within
  it.  Moreover, the endogenous covariates $X_{ij}$ for
  $i,j\in t_s$ depend only on link outcomes incident to $i$ or~$j$
  (Assumption~\ref{ass:local_ext}).  Any such link $(i,k)$ with
  $k\notin D_{t_s}$ is robustly absent---since $k$ is not in any
  strategic neighborhood of a member of~$t_s$, we have $Y_{ik}=0$---so
  $X_{ij}$ depends only on primitives within $D_{t_s}$.
  Therefore $\{\ind\{E_{t_s}\}\}_{s=1}^{m_n}$ are mutually
  independent.

  Moreover, because
  $\tilde Z_i=(Z_i,A_i)$ is i.i.d.\ across agents
  (Assumption~\ref{ass:random_sampling}), any two tetrads with
  $Z_S=z_S$ have the same conditional distribution of outcomes.
  The identical-mean property holds because each $E_{t_s}$ depends
  only on primitives within~$D_{t_s}$, and conditional on
  $Z_{S,t_s}=z_S$, the distribution of these primitives is
  determined by the i.i.d.\ structure of $(\tilde Z_i,\varepsilon_{ij})$
  and does not depend on which tetrad was selected.
  In particular,
  \[
    \mathbb{E}\bigl[\ind\{E_{t_s}\}\mid Z_{S,t_s}=z_S\bigr]
    \;=\;
    p^{(n)}(z_S)
    \quad\text{for all } s=1,\ldots,m_n.
  \]

  Since $\{\ind\{E_{t_s}\}\}_{s=1}^{m_n}$ are i.i.d.\ Bernoulli
  random variables (conditional on $Z_{S,t_s}=z_S$ for all~$s$)
  and $m_n\to\infty$, the strong law of large numbers gives
  \[
    \frac{1}{m_n}\sum_{s=1}^{m_n}\ind\{E_{t_s}\}
    \;\xrightarrow{\;p\;}\;
    p^{(n)}(z_S).
  \]
  This establishes that an oracle with access to the packed tetrads
  could consistently estimate $p^{(n)}(z_S)$.  Showing that the
  natural full-sample average over all observable matched tetrads
  (a U-statistic) is also consistent requires additional arguments
  controlling the contribution of dependent tetrad pairs; we leave
  this formal extension to future work.
\end{proof}

\subsection{Proof of Theorem~\ref{thm:point_id}}

\begin{proof}
  Write $\Lambda(x):=(1+e^{-x})^{-1}$ for the logistic CDF.
  Fix an admissible tetrad $(i,j,h,k)\in\mathcal{T}_n$.

  \medskip\noindent\emph{Step~1: Conditional independence of tetrad shocks.}
  By Assumption~\ref{ass:tetrad_exog},
  $Y_{-E_t}\in\sigma(Z,A,\varepsilon_{-E_t})$.  Since
  $\{\varepsilon_{ij}\}_{i<j}$ are i.i.d.\
  $\mathrm{Logistic}(0,1)$ independently of $(Z,A)$, we have
  $\varepsilon_{E_t}\perp\sigma(Z,A,\varepsilon_{-E_t})$.
  Because
  $\sigma(Z,A,Y_{-E_t})\subseteq\sigma(Z,A,\varepsilon_{-E_t})$,
  it follows that conditional on $(Z,A,Y_{-E_t})$, the four tetrad
  shocks remain i.i.d.\ $\mathrm{Logistic}(0,1)$.

  \medskip\noindent\emph{Step~2: Best-response conditions on $T_t$
  and $F_t$.}
  By Assumption~\ref{ass:cpi}, on both $T_t$ and $F_t$ (with
  $Y_{ih}=Y_{jk}=0$), each $X_e$ equals
  $\tilde\phi_e(Y_{-E_t},Z)$, which is determined by
  $(Z,A,Y_{-E_t})$.
  Define
  $\tilde\delta_e:=Z_e'\beta_0+\tilde\phi_e(Y_{-E_t},Z)'\gamma_0$
  and
  $\tilde\delta_e^A:=\tilde\delta_e-A_{i(e)}-A_{j(e)}$.
  The best-response conditions for~$T_t$ to be an equilibrium are
  \[
    \varepsilon_{ij}\le\tilde\delta_{ij}^A,\quad
    \varepsilon_{hk}\le\tilde\delta_{hk}^A,\quad
    \varepsilon_{ik}>\tilde\delta_{ik}^A,\quad
    \varepsilon_{jh}>\tilde\delta_{jh}^A.
  \]
  Symmetrically, the best-response conditions for~$F_t$ reverse
  every inequality.  Because the same thresholds
  $\tilde\delta_e^A$ appear in both patterns (by
  comparison-pattern invariance), the two events occupy disjoint
  rectangles in~$\R^4$.

  \medskip\noindent\emph{Step~3: Odds ratio.}
  By Assumption~\ref{ass:local_uniq}, whenever the shocks fall in
  the rectangle for~$T_t$ (resp.~$F_t$), the tetrad pattern~$T_t$
  (resp.~$F_t$) is the unique equilibrium.  Since the tetrad shocks
  are i.i.d.\ $\mathrm{Logistic}(0,1)$ conditional on
  $(Z,A,Y_{-E_t})$:
  \begin{align*}
    \frac{\Prob(T_t\mid Z,A,Y_{-E_t})}
    {\Prob(F_t\mid Z,A,Y_{-E_t})}
    &=\prod_{e\in\{ij,hk\}}
      \frac{\Lambda(\tilde\delta_e^A)}{1-\Lambda(\tilde\delta_e^A)}
      \cdot\prod_{e\in\{ik,jh\}}
      \frac{1-\Lambda(\tilde\delta_e^A)}{\Lambda(\tilde\delta_e^A)}\\
    &=\exp\!\Bigl[\sum_{e\in\{ij,hk\}}\tilde\delta_e^A
      -\sum_{e\in\{ik,jh\}}\tilde\delta_e^A\Bigr].
  \end{align*}
  The fixed-effect terms cancel:
  \[
    (A_i+A_j)+(A_h+A_k)-(A_i+A_k)-(A_j+A_h)=0,
  \]
  giving
  \[
    \frac{\Prob(T_t\mid Z,A,Y_{-E_t})}
    {\Prob(F_t\mid Z,A,Y_{-E_t})}
    =\exp\!\bigl(\Delta Z'\beta_0+\Delta\tilde X'\gamma_0\bigr),
  \]
  where $\Delta\tilde X:=\tilde\phi_{ij}+\tilde\phi_{hk}
  -\tilde\phi_{ik}-\tilde\phi_{jh}$.
  The ratio does not depend on~$A$.

  \medskip\noindent\emph{Step~4: Integration.}
  Since $\exp(\Delta Z'\beta_0+\Delta\tilde X'\gamma_0)$ depends
  on $(Z,A,Y_{-E_t})$ only through $(\zeta,\tilde X_t)$---and on
  $T_t\cup F_t$ we have $\tilde X_t=X_t$ by
  Assumption~\ref{ass:cpi}---the ratio is constant given
  $(\zeta,x_t)$ and factors out when taking expectations over
  $(A,Y_{-E_t})$ conditional on
  $(\zeta,X_t=x_t,Y_{ih}=0,Y_{jk}=0)$:
  \begin{align*}
    p_+(\zeta,x_t)
    &=\E\!\bigl[\Prob(T_t\mid Z,A,Y_{-E_t})
    \bigm|\zeta,x_t,Y_{ih}=0,Y_{jk}=0\bigr]\\
    &=\exp(\Delta Z'\beta_0+\Delta X'\gamma_0)\cdot
    \E\!\bigl[\Prob(F_t\mid Z,A,Y_{-E_t})
    \bigm|\zeta,x_t,Y_{ih}=0,Y_{jk}=0\bigr]\\
    &=\exp(\Delta Z'\beta_0+\Delta X'\gamma_0)\cdot p_-(\zeta,x_t).
  \end{align*}
  Taking logs yields~\eqref{eq:logodds_fe}.

  The equation
  $\log[p_+(\zeta,x_t)/p_-(\zeta,x_t)]=(\Delta Z,\Delta X)'
  (\beta_0,\gamma_0)$ is a system of moment equalities indexed by
  $(\zeta,x_t)$.  Under the rank condition, this system has a unique
  solution.
\end{proof}

\end{document}